\documentclass[acmsmall, screen]
{acmart}

\AtBeginDocument{%
  }

\setcopyright{none}

 \acmJournal{PACMPL}
\acmVolume{1}
\acmNumber{CONF} 
\acmArticle{1}
\acmYear{2018}
\acmMonth{1}
\acmDOI{} 
\startPage{1}

\usepackage[utf8]{inputenc}
\usepackage[T1]{fontenc}
\usepackage{microtype}
\usepackage{booktabs}   
\usepackage{subcaption} 
\usepackage{braket}
\usepackage{proof}
\usepackage{pgfplots} \pgfplotsset{compat=1.7}
\usepackage[normalem]{ulem}
\usepackage{tikz}
\usepackage{graphicx}
\usetikzlibrary{arrows,positioning}
\usepackage{wrapfig}
\usepackage{csquotes}
\usepackage{etoolbox}
\usepackage{mathpartir} 
\usepackage{listings}
\usepackage{xcolor}
\usepackage{makecell}
\usepackage{amsmath}
\usepackage{enumitem}

\usepackage{cleveref}
\crefname{section}{\S}{\S\S}
\Crefname{section}{\S}{\S\S}
\crefformat{section}{\S#2#1#3}

\usepackage{hyperref}
\usepackage{upgreek}
\usepackage{algorithm}
\usepackage{algpseudocode}
\usepackage[framemethod=default]{mdframed}
\mdfsetup{skipabove=.67\topskip,skipbelow=.67\topskip}
\usepackage{adjustbox}
\usetikzlibrary{quantikz}
\usepackage{multirow}
\newcommand {\cH } {{\mathcal{H}}}
\newcommand {\ass }{{\bf assert}}
\newcommand {\supp } {{\rm supp}}
\newcommand {\Tr} {{\mathrm{Tr}}}
\newcommand {\tr} {{\mathrm{Tr}}}

\newcommand{\op}[2]{|#1\rangle \langle #2|}
\newcommand{\nm}[1]{\lVert #1\rVert}
\newcommand{\ZZ}{\mathbb{Z}}
\newcommand{\FF}{\mathbb{F}}
\newtheorem{thm}{Theorem}[section]

\newtheorem{lem}{Lemma}[section]
\newtoggle{NY}
\newtheorem{exam}{Example}[section]
\newcommand{\twr}[1]{{\color{magenta}{T: #1}}}
\newcommand{\twrchanged}[1]{{\color{cyan}{#1}}}

\usepackage{wrapfig}
\usepackage{lipsum}
\begin{document}
\title[Scalable Equivalence Checking and Verification of Shallow Quantum Circuits]{Scalable Equivalence Checking and Verification \\ of Shallow Quantum Circuits}
\author{Nengkun Yu}
\orcid{0000-0003-1188-3032}             
\affiliation{             
  \institution{Stony Brook University}            
  \country{USA}                    
}

\author{Xuan Du Trinh}         
\affiliation{             
  \institution{Stony Brook University}            
  \country{USA}                    
}

\author{Thomas Reps}
\orcid{0000-0002-5676-9949}             
\affiliation{             
  \institution{University of Wisconsin–Madison}            
  \country{USA}                    
}

\begin{abstract}

This paper concerns the problem of checking if two shallow (i.e., constant-depth) quantum circuits perform equivalent computations.
Equivalence checking is a fundamental correctness question---needed, e.g., for ensuring that transformations applied to a quantum circuit do not alter its behavior.
For quantum circuits, the problem is challenging because a straightforward representation on a classical computer of each circuit's quantum state can require time and space that are exponential in the number of qubits $n$.

The paper presents decision procedures for two variants of the equivalence-checking problem.
Both can be carried out on a classical computer in time and space that, for any fixed depth, is linear
in $n$.
Our critical insight is that local projections are precise enough to completely characterize the output state of a shallow quantum circuit.
Instead of explicitly computing the output state of a circuit, we generate a set of local projections that serve as constraints on the output state.
Moreover, the circuit's output state is the unique quantum state that satisfies all the constraints.

Beyond equivalence checking, we show how to use the constraint representation to check a class of assertions, both statically and at run time.
Our assertion-checking methods are sound and complete for assertions expressed as conjunctions of local projections.

Our experiments show that
on a server equipped with 2× Intel\textsuperscript{\textregistered} Xeon\textsuperscript{\textregistered} Gold 6338 CPUs (128 threads total) and 1.0~TiB of RAM, running Ubuntu 20.04.6 LTS, the constraint representation of a random 100-qubit circuit of depth 6 can be computed in 19.8 seconds.
For fixed inputs $\ket{0}^{\otimes 100}$, equivalence checking of {random} 100-qubit circuits of depth 3 takes 4.46 seconds;
for arbitrary inputs, it takes no more than 31.96 seconds.

\end{abstract}

\maketitle

%
%
%

\section{Introduction}
\label{Se:Introduction}

To make programming quantum computers easier, researchers are developing quantum programming languages \cite{SZ00,Sa03,O03,Se04,AG05,DBLP:conf/pldi/BichselBGV20} and platforms for implementing quantum software \cite{JFD12,GLR13,WS14,Svor18,Cirq,Qiskit}.
Verifying quantum software on classical computers is both essential and challenging, given the noisy and nascent state of quantum hardware 
\cite{BHY19,ZYY19,10.1145/3498697,Unruh,Unruh2,YP21,DBLP:journals/pacmpl/Hietala0HW021,Peng2023,DBLP:journals/pacmpl/YuanMC22,DBLP:conf/pldi/TaoSYLJCCG22,Ying11,Bichsel_2023,10.1145/3527316}.
Verification techniques originally designed for software can also support validation of quantum hardware \cite{li2019proq}.
Among these advances, relational analysis for quantum programs has gained a certain amount of attention \cite{BHY19,li2019quantum,Unruh,approximate_relational_reasoning}.
One of the most important relational properties is program equivalence \cite{10.1007/3-540-09526-8-20,BERGSTRA1982113}---needed, for instance, for ensuring that transformations applied to a quantum circuit do not alter its behavior.


Equivalence checking is a fundamental concept in computer science \cite{10.1007/3-540-09526-8-20,BERGSTRA1982113}, with applications in areas such as electronic design automation, translation validation, and program optimization.
For quantum circuits---one of the most important classes of quantum programs---the equivalence problem has been the subject of several studies \cite{10.5555/2011464.2011465,10.5555/1326073.1326089,BURGHOLZER2021100051,10.1145/3508352.3549479,DBLP:conf/atva/ThanosCL23,Amy_2019,DBLP:conf/ijcar/MeiCBL24}.
\citet{10.1145/3508352.3549479} proposed two decision-diagram-based algorithms to check the equivalence of dynamic quantum circuits.
Amy \cite{Amy_2019} and \citet{DBLP:conf/atva/ThanosCL23} proposed efficient algorithms for checking equivalence in Clifford-gate quantum circuits.
\citet{DBLP:conf/ijcar/MeiCBL24} developed a precise method to check the equivalence of universal quantum circuits using weighted-model-counting techniques. 

However, it is challenging to develop an \textit{efficient} equivalence-checking technique for quantum circuits.
From one perspective, simulation is relevant:
many quantum algorithms start with a single known input (e.g., $\ket{0}^{\otimes n} = \ket{0 \ldots 0}$);
simulation computes the output state by mimicking the computation step-by-step;
and thus one way to establish whether the computation of two quantum circuits are equivalent would be by comparing the two output states.
Unfortunately, with all presently known classical simulation techniques, in the general case the time and space costs scale exponentially.\footnote{
  Classical simulation refers to a classical algorithm that computes (or in some cases, approximates) the entire output state of a circuit, or can sample from the output state according to the exact or approximate amplitudes \cite{10.1145/3618260.3649638,PhysRevX.12.021021}.
}
An alternative approach is to confine attention to a limited class of quantum circuits.
For instance, \citet{DBLP:conf/atva/ThanosCL23}, exploit the Knill-Gottesman theorem \cite{gottesman1998heisenbergrepresentationquantumcomputers,Aaronson_2004}---which enables efficient classical simulation of Clifford circuits---to develop an efficient equivalence-checking method for Clifford circuits.
Amy applied his Feynman path-integral-based circuit-verification methods to Clifford+T circuits, but their efficiency remains theoretically uncertain \cite{Amy_2019}.

In our work, while we also confine our attention to a limited class of quantum circuits $\mathcal{C}$, our goal is to find an equivalence-checking method for $\mathcal{C}$ that is (i) efficient---i.e., polynomial time and space in the number of qubits $n$---but where (ii) $\mathcal{C}$ is known to contain families of circuits for which classical simulation is intractable.
\noindent
\begin{mdframed}
  \textit{Is it feasible to check equivalence for a class of quantum circuits that contains families of circuits for which classical simulation is intractable?}
\end{mdframed}
\noindent
{\it A priori}, it is unclear whether this goal is achievable.
Conventional wisdom suggests it is not \cite{sigarch},
the intuition being something like
(i) equivalence checking is a form of verification;
(ii) verification methods are a form of state-space exploration;
(iii) state-space exploration is a form of simulation;
(iv) if equivalence checking is tractable, then classical simulation is tractable---which directly contradicts our desire to handle families of circuits for which classical simulation is intractable!

In our work, we focus on the class of \emph{constant-depth} quantum circuits, where the circuit depth $d$ is left as an unspecified constant.\footnote{
  When establishing a complexity bound, a family of quantum circuits is parameterized on the number of qubits $n$.
  Our results apply to families where the quantum-circuit depth $d$ is independent of $n$.
  They are also applicable to $n$-qubit 1D circuits of depth $O(\log n)$.
}
Because each member has fixed depth $d$, constant-depth circuits run in constant time on a quantum computer.
Constant-depth quantum circuits have been shown to outperform their classical counterparts, even in the presence of noise \cite{Bravyi_2018,Bravyi_2020}.
Moreover, these circuits are particularly well-suited for implementation in the Noisy Intermediate-Scale Quantum (NISQ) era \cite{preskill2018quantum}, because a shallow-depth circuit only requires coherence times that are achievable with present-day technology.
Very recently, \citet{schuster2025randomunitariesextremelylow} demonstrated that shallow-depth random quantum circuits can serve as a foundation for quantum cryptographic primitives.\footnote{
  This result shows that for 1D system, $O(\log n)$ depth is enough.
}
Thus, constant-depth (shallow) circuits provide a promising route for achieving quantum supremacy in the NISQ era \cite{preskill2018quantum}, as Google demonstrated with Sycamore \cite{Arute_2019, morvan2023phasetransitionrandomcircuit}.
(Quantum supremacy means that a quantum computer can perform a task that is intractable to a classical computer.) 

On the other hand, constant-depth circuits are, in general, hard to simulate classically:
\citet{terhal2004adaptivequantumcomputationconstant} provide evidence that there are quantum computations that
(i) can be performed by a constant-depth circuit using 2-qubit gates, but
(ii) cannot be accurately simulated classically.\footnote{
  \citeauthor{terhal2004adaptivequantumcomputationconstant} study depth-4 circuits;
  however, the final layer consists solely of measurements on a computational basis.
  Thus, their results indicate that it is difficult to simulate measurement-free quantum circuits of depth 3.
} 
They also demonstrate that efficient classical simulation of these circuits, up to a constant precision, would imply that the complexity class BQP is contained within AM.
\citet{ji2009nonidentitycheckremainsqmacomplete} proved
that approximate equivalence-checking for constant-depth circuits is QMA-hard,
given at last $\Omega(\log n)$ bits of precision for each gate.
(BQP and QMA are the quantum analogs of P and NP, respectively.)

\paragraph{\textbf{Our contributions}}
In light of the negative results discussed above, the prospects for being able to do something classically for constant-depth quantum circuits look rather bleak.
However, in this paper, we show that
\noindent
\begin{mdframed}
  \textit{One can check the equivalence of two constant-depth quantum circuits efficiently, in time linear in the number of qubits $n$.}
\end{mdframed}
\noindent
From the perspective of computational complexity, the intuition behind our result is that the costs of our methods exhibit \emph{fixed-parameter tractability} \cite{DBLP:journals/siamcomp/DowneyF95}, with time complexity of the form $T(n,d) = f(d) \cdot O(n)$, where $d$ is the circuit depth, and $f(d)$ is a function that captures constraints on qubit interactions imposed by the circuit’s geometry (e.g., 1D or 2D architectures).\footnote{
  A similar fixed-parameter-tractability argument explains why LTL model checking is tractable in practice.
  The complexity is $O(2^{|\varphi|} \cdot |M|)$, where $\varphi$ is the formula and $M$ is the model.
  However, most formulas of interest are small, so the exponential term $2^{|\varphi|}$ is effectively a constant for practically relevant formulas, and thus in practice an LTL model-checking problem can be solved in linear time: $O(M)$.
}
For 1D architectures, $f(d) = 2^{O(d)}$, while for 2D architectures, $f(d) = 2^{O(d^2)}$.
Thus, $T(n,d)$ is exponential in $d$, but for a fixed value of $d$, $T(n,d)$ is linear in $n$, enabling efficient scaling with problem size.

We assume no restrictions on the gate set, other than the property that each elementary gate acts on a constant number of qubits.
Throughout the rest of the paper, we use the phrase ``shallow circuit'' as a synonym for ``constant-depth circuit.''
Our examples use 1- and 2-qubit gates, but our results apply to circuits composed of 3-qubit gates (such as the Toffoli gate), 4-qubit gates, etc.

\paragraph{A constraint-based description of a circuit's output state (\Cref{Se:Overview} and \Cref{Se:AnEfficientDescription})}

The key insight underlying our result is that classical simulation is \emph{not required} to be able to compare the outputs of two shallow circuits.
Instead of using simulation, we give an algorithm that provides an exact \emph{specification} of the result computed by a shallow quantum circuit applied to the initial state $\ket{0}^{\otimes n}$.
As is common in many quantum algorithms \cite{deutsch1992rapid, deutsch1992rapid, Arute_2019, morvan2023phasetransitionrandomcircuit}, the computation defined by such a circuit produces a unique output quantum state.
The specification---or \emph{description}---of the circuit's output state is captured as a \emph{tuple of local projections}, similar to the abstract state in quantum abstract interpretation \cite{YP21}.
Essentially, such a tuple can be viewed as a conjunction of atomic constraints (i.e., the individual local projections).
The most notable features of this method are:
\begin{description}
 \item [Completeness:]
   The constraint description produced for a given constant-dept circuit specifies exactly the circuit's output state.

  \item [Efficient description size:] For an $n$-qubit constant-depth circuit, the constraint description consists of at most $n$ local projections, where each local projection depends non-trivially on only a constant number of qubits.

  \item [Efficiently computable on a classical computer:] Using a classical computer, a constraint description can be created in the time polynomial in the number of qubits $n$.
\end{description}

\noindent
These features help us address several problems in quantum program analysis, as described below.

\paragraph{Efficient equivalence checking (\Cref{Se:EquivalenceChecking})}
The constraint description of a quantum circuit's output enables us to give algorithms for two variants of quantum-circuit equivalence (for constant-depth quantum circuits): 
\begin{enumerate}[left=0pt .. 1.5\parindent]
  \item
    Do two circuits produce equivalent output states when applied to the initial state $\ket{0}^{\otimes n}$?
  \item
    Do two circuits produce equivalent output states for each possible initial state $\ket{\psi}$?
\end{enumerate}

\noindent
The first algorithm checks the equivalence of quantum algorithms whose inputs are fixed as $\ket{0}^{\otimes n}$.
When an optimizing transformation is applied to a quantum circuit, the second algorithm can be used to verify the equivalence of the original and transformed circuits.

Note that the first notion of equivalence is a weaker notion than the second.
We address the first kind of equivalence-checking problem as our initial task because (i) the initial state of most quantum algorithms is $\ket{0}^{\otimes n}$, and (ii) our approach to the second kind of equivalence-checking problem is to reduce it to a problem of the first kind.

Even for the first kind of equivalence-checking problem, it is not feasible to compare the respective local-projection tuples directly because there is no canonical form that would make equivalence-checking easy.\footnote{
  Similarly, it is not feasible to perform equivalence checking by comparing the respective circuit structures directly:
  equivalent circuits can have entirely different structures, and we do not know of a canonical form for constant-depth quantum circuits.
  Nor do we know of a complete set of rewrite rules by which a given circuit can be converted into every equivalent circuit.
}
For this variant of equivalence checking, we leverage the reversibility of quantum circuits to transform the problem into an identity-checking problem.
For the second equivalence-checking problem, we use Choi states \cite{CHOI1975285} to reduce the second kind of problem to the first kind.

\paragraph{Applications to assertion checking (\Cref{Se:AssertionChecking})}
Beyond equivalence checking, we
show how the constraint description of the output of a quantum circuit enables the verification of a class of assertions both statically and at runtime.
In both cases, the assertion language $\mathcal{L}$ consists of conjunctions of local projections, and our assertion-checking methods are sound and complete for $\mathcal{L}$.

For static assertion checking, we show that 
\begin{mdframed}
  \textit{One can efficiently check an $\mathcal{L}$-assertion for any constant-depth quantum circuit.}
\end{mdframed}

Any approach to runtime assertion checking for quantum computing faces three issues:
\begin{enumerate}
  \item
    Can each assertion be expressed efficiently?
  \item
    Can each assertion be implemented efficiently on a quantum computer?
  \item
    Can each assertion be checked without affecting the state of the computation if the assertion is satisfied?
\end{enumerate}

Assertions based on local projections provide an affirmative answer to all three questions:
\begin{description}
  \item [Expressiveness:]
    A conjunction of $n$ local projections can describe exactly any intermediate state of the computation of a constant-depth quantum circuit.\footnote{
      We do not suggest that conjunctions of local projections are intuitive for human users.
      Whether a ``human-friendly'' assertion language can be compiled into local projections as an intermediate language remains for future research.
    }

  \item [Efficiently implementable:]
    Assertion checking can be implemented efficiently as a sequence of projective measurements.
    Each projection is local, acting on a constant number of qubits, making them easy to implement.
 The local projections commute, so performing the corresponding measurements in different orders produces the same binary measurement.

  \item [Non-intrusiveness:]
    The assertion language of local projections inherits the advantageous property of the language of general projective assertions \cite{li2019proq}, namely, the presence of an assertion check does not affect the state of a computation if the assertion is satisfied.
\end{description}

\paragraph{Organization of the paper}
\Cref{Se:BackgroundOnQuantumComputing} summarizes the basic concepts of quantum computing.
\Cref{Se:Overview} discusses an example to show how we create a constraint-based description of the output state of a constant-depth quantum circuit.
\Cref{Se:AdditionalTerminologyAndNotation} explains some additional technical concepts needed in the remainder of the paper.
\Cref{Se:AnEfficientDescription} gives our algorithm for computing an 
efficient constraint-based description of a circuit's output state.
\Cref{Se:EquivalenceChecking} presents our method for efficient equivalence checking.
\Cref{Se:AssertionChecking} presents techniques for checking local-projection assertions, both statically and at runtime.
\Cref{sec:Experiments} describes experiments with an implementation of the equivalence-checking technique.
\Cref{Se:RelatedWork} discusses related work.
\Cref{sec:Conclusion} concludes.
{We defer some proofs to \Cref{sec:ProofOfTheorem}, and give a comprehensive comparison with quantum abstract interpretation \cite{YP21} in \Cref{Se:ComparisonWithQuantumAbstractInterpretation}.}
\Cref{sec:ProofOfTheorem} and \Cref{Se:ComparisonWithQuantumAbstractInterpretation} are submitted as Supplementary Material.

\section{Background on Quantum Computing}
\label{Se:BackgroundOnQuantumComputing}

This section presents background about, and notation used in, quantum information and quantum computation, primarily following the textbook by Nielsen and Chuang~\citeN{NI11}.

\paragraph{Notation}

We use the notation
$[n]=\{1,2,\cdots,n\}$
and ``$\setminus$'' to denote set difference.
The cardinality of a set $s$ is denoted by $|s|$.
We focus on finite-dimensional vector spaces $\mathbb{C}^d$ of complex vectors.
Linear \emph{operators} are linear mappings between these vector spaces, represented by $d \times d$ matrices, denoted by $\mathbb{C}^{d \times d}$.
The identity matrix is denoted by $I$.
The Hermitian conjugate of an operator $A$ is $A^\dag = (A^T)^*$, where $A^T$ is the transpose of $A$, and $B^*$ is the complex conjugate of $B$.
An operator $A$ is \emph{Hermitian} if $A = A^\dag$.
A Hermitian operator $A$ is positive semi-definite if it has non-negative eigenvalues.
The trace of a matrix $A$ is the sum of its diagonal entries, $\tr(A) = \sum_i A_{ii}$. 

We assume familiarity with linear-algebra concepts, such as tensor products, orthonormal bases, inner products, outer products, and Hilbert spaces. We use Dirac notation, $\ket{\psi}$, to denote a complex column vector in $\mathbb{C}^d$. The inner product of vectors $\ket{\psi}$ and $\ket{\phi}$ is $\langle\psi|\phi\rangle \in \mathbb{C}$, which is the matrix product of $\bra{\psi}$, the Hermitian conjugate of $\ket{\psi}$, and $\ket{\phi}$. The outer product of vectors $\ket{\psi}$ and $\ket{\phi}$ is $\op{\psi}{\phi} \in \mathbb{C}^{d \times d}$, the matrix product of $\ket{\psi}$ and $\bra{\phi}$, the Hermitian conjugate of $\ket{\phi}$. The Euclidean norm of a vector $\ket{\psi}$ is $\nm{\ket{\psi}} = \sqrt{\langle\psi|\psi\rangle}$.

\paragraph{Quantum States}

Quantum states are the fundamental descriptors of quantum systems, providing a complete and probabilistic description of a system's properties.
A (pure) quantum state is represented by a vector \(|\psi\rangle\) in a Hilbert space, which can be expressed as a superposition of basis states, encapsulating the principle of superposition inherent in quantum mechanics. 
When dealing with multi-qubit systems, the overall quantum state resides in a higher-dimensional Hilbert space formed by the tensor product of the individual qubit spaces.
The state space grows exponentially in the number of qubits.




\paragraph{Unitary Operations}

In quantum mechanics, unitary operators are fundamental transformations that play a crucial role in preserving the quantum-mechanical properties of systems.
A unitary operator \(U\) is represented by a matrix that satisfies the condition \(U^\dagger U = U U^\dagger = I\), where \(U^\dagger\) denotes the Hermitian adjoint (conjugate transpose) of \(U\), and \(I\) is the identity matrix.
This condition ensures that unitary transformations are reversible and conserve the norm of quantum states, thereby preserving probabilities. We call one or two-qubit unitary matrices
``quantum gates.''

\paragraph{Quantum Circuits}

Quantum circuits are a formalism for expressing quantum algorithms.
In a quantum circuit, quantum gates are organized in layers, with each layer containing gates that act on different qubits simultaneously. 
To keep things simple, and to reduce notational clutter, we formulate the semantics of a quantum circuit for the case of 2-qubit gates.
However, our results apply to circuits composed of gates that act on up to $m$ qubits, where $m$ is a constant.

The semantics of each layer can be expressed as a tensor product of 2-qubit unitary matrices.

Let \( U_{ij}^{(k)} \) denote the 2-qubit unitary acting on qubits \( i \) and \( j \) in circuit layer \( k \).
The general evolution of the quantum state after applying the \( L \) layers of such operations can be expressed as
\begin{equation}
\label{Eq:QuantumCircuitSemantics}
  |\psi_{\text{final}}\rangle = \left( \prod_{k=1}^{L} \bigotimes_{(i,j) \in \text{pairs}[k]} U_{ij}^{(k)} \right) |\psi_{\text{initial}}\rangle,
\end{equation}
where $(i,j) \in \text{pairs}[k]$ denotes the set of pairs of qubits acted upon by a unitary \( U_{ij}^{(k)} \) in layer \( k \), and $|\psi_{\text{initial}}\rangle$ and $|\psi_{\text{final}}\rangle$ represent the initial and final quantum states, respectively.
One often has $|\psi_{\text{initial}}\rangle := |0^{\otimes n}\rangle$.
Each unitary \( U_{ij}^{(k)} \) is a \( 4 \times 4 \) unitary matrix representing the operation on qubits \( i \) and \( j \) in layer \( k \).
Here, we are padding 1-qubit \( 2 \times 2 \) unitary matrices to 2-qubit \( 4 \times 4 \) unitary matrices.
The indexed tensor product $\bigotimes_{(i,j) \in \text{pairs}[k]} U_{ij}^{(k)}$ in \Cref{Eq:QuantumCircuitSemantics} denotes the semantics of layer $k$ of the circuit.
The indexed product $\left( \prod_{k=1}^{L} \cdot \right)$ denotes the composition of the $L$ layers.
The depth of a quantum circuit refers to the number of layers $L$ that it contains, representing the application of gates over $L$ time steps.

\section{Overview}%
\label{Se:Overview}

\begin{wrapfigure}{R}{7.0cm}
    \vspace{-2.0ex}
    \centering
    \begin{adjustbox}{width=0.5\textwidth}
    \begin{quantikz}[column sep={0.5cm}, row sep={0.2cm}]
        \lstick{$q_1\ket{0}$} & \gate[style={fill=blue!20}, wires=2][1cm]{U_{1,2}^{(1)}} & \qw & \gate[style={fill=blue!40}, wires=2][1cm]{U_{1,2}^{(3)}} & \qw & \qw \\
        \lstick{$q_2\ket{0}$} & \qw & \gate[style={fill=blue!30}, wires=2][1cm]{U_{2,3}^{(2)}} & \qw &\qw & \qw \\ 
        \lstick{$q_3\ket{0}$} & \gate[style={fill=blue!20}, wires=2][1cm]{U_{3,4}^{(1)}} & \qw & \gate[style={fill=blue!40}, wires=2][1cm]{U_{3,4}^{(3)}} & \qw & \qw \\
        \lstick{$q_4\ket{0}$} & \qw & \gate[style={fill=blue!30}, wires=2][1cm]{U_{4,5}^{(2)}} & \qw & \qw& \qw \\
        \lstick{$q_5\ket{0}$} & \gate[style={fill=blue!20}, wires=2][1cm]{U_{5,6}^{(1)}} & \qw & \gate[style={fill=blue!40}, wires=2][1cm]{U_{5,6}^{(3)}} & \qw & \qw \\
        \lstick{$q_6\ket{0}$} & \qw & \gate[style={fill=blue!30}, wires=2][1cm]{U_{6,7}^{(2)}} & \qw & \qw & \qw \\
        \lstick{$q_7\ket{0}$} & \gate[style={fill=blue!20}, wires=2][1cm]{U_{7,8}^{(1)}} & \qw & \gate[style={fill=blue!40}, wires=2][1cm]{U_{7,8}^{(3)}} & \qw & \qw \\
        \lstick{$q_8\ket{0}$} & \qw & \qw & \qw & \qw & \qw 
    \end{quantikz}
    \end{adjustbox}
    \caption{The example: each layer consists of two-qubit unitaries applied to disjoint qubits.}
    \label{fig:circuit_1}
    \vspace{-3.0ex}
\end{wrapfigure}
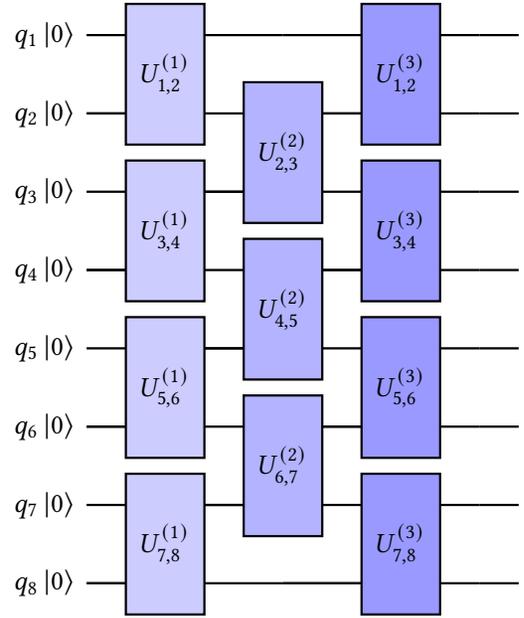

This section presents an example to show how we create a constraint-based description of the output state of a constant-depth quantum circuit, as a tuple of local projections.
This representation is what unlocks our results on equivalence checking (\Cref{Se:EquivalenceChecking}) and assertion checking (\Cref{Se:AssertionChecking}). 
More precisely, the obtained constraints are local, meaning that each constraint acts nontrivially on only a constant number of qubits.
This locality enables efficient equivalence checking and assertion verification, making both tasks computationally feasible.


\Cref{fig:circuit_1} shows an example a depth-3 quantum circuit. Each of the eight horizontal rows in the diagram depicts a qubit.
Time progresses from left to right, with the initial value of each qubit being \( \ket{0} \), as shown on the left.
The qubit \( \ket{0} \) is written in Dirac notation and can also be expressed as the column vector
$\begin{pmatrix}
1 \ \ 
0
\end{pmatrix}^t$. In quantum computing, the analogs of logic gates are quantum \emph{gates}, which are mathematically represented by matrices. Consequently, we will use the two terms interchangeably.

\paragraph{Infeasibility of classical simulation-based verification}
Because this circuit could contain non-Clifford unitary matrices, the standard classical simulation algorithm, based on the famous Gottesman-Knill Theorem \cite{gottesman1998heisenbergrepresentationquantumcomputers,Aaronson_2004}, is not applicable here.

We will walk through the example, state the challenge of the paper
and illustrate the main idea of our method.
Along the way, we will recall key concepts of quantum computing.

The initial state is $\ket{0^{\otimes 8}}$.
The unitary matrices of the three layers---working left to right---are
\[
U^{(1)} = U_{1,2}^{(1)}\otimes {U_{3,4}^{(1)}}  \otimes{U_{5,6}^{(1)}} \otimes {U_{7,8}^{(1)}}  \quad
U^{(2)} = I_1\otimes U_{2,3}^{(2)}\otimes {U_{4,5}^{(2)}}  \otimes{U_{6,7}^{(2)}} \otimes  I_{8} \quad
U^{(3)} = U_{1,2}^{(3)}\otimes {U_{3,4}^{(3)}}  \otimes{U_{5,6}^{(3)}} \otimes {U_{7,8}^{(3)}} 
\]
with $I_1$ and $I_{8}$ being the identity matrices on qubit $q_1$ and $q_{8}$, respectively.
For brevity, we may occasionally omit such identity matrices in our presentation.
The unitaries \( U_{i,j}^{(k)} \) are general \( 4 \times 4 \) unitary matrices with no additional restrictions imposed on them.

The output state of this circuit is $U^{(3)}U^{(2)}U^{(1)}\ket{0^{\otimes 8}}$, where juxtaposition denotes matrix multiplication (or, a special case, matrix-vector multiplication).

Before looking at the output state of this circuit, consider the initial state $\ket{\Phi_0}:=\ket{0^{\otimes 8}}$.
On the one hand, it can be represented as a $2^{8}$-dimensional complex vector.
However, we can have a more compact representation:
the first qubit is in state $\ket{0}$; 
the second qubit is in state $\ket{0}$;
$\cdots$;
the $8^{\textit{th}}$ qubit is in state $\ket{0}$.
In other words, this state can be represented as a tuple, a conjunction,
\[
\begin{array}{l}(\ket{0},\ket{0},\ket{0},\ket{0},\ket{0},\ket{0},\ket{0},\ket{0})
\end{array}
\]
where the $i^{\textit{th}}$ entry denotes the state of the $i^{\textit{th}}$ qubit.

Now consider the result of the first step:
after applying the first layer of unitaries, we have
\[
  \ket{\Phi_1}:=U^{(1)}\ket{0^{\otimes 8}}=U_{1,2}^{(1)}\otimes {U_{3,4}^{(1)}}  \otimes{U_{5,6}^{(1)}} \otimes {U_{7,8}^{(1)}} \ket{0^{\otimes 8}}.
\]
Unfortunately, for general choices of unitaries, this state is no longer a tensor product of one-qubit states. In other words, the action of $U^{(1)}$ can generate quantum entanglement. However, it is still in a tensor-product form,
\begin{align*}
  U^{(1)}\ket{0^{\otimes 8}}=U_{1,2}^{(1)}\ket{00}\otimes {U_{3,4}^{(1)}}  \ket{00}\otimes{U_{5,6}^{(1)}} \ket{00}\otimes {U_{7,8}^{(1)}}\ket{00}.
\end{align*}
It is a tensor product of two-qubit states instead of one-qubit states.
In other words, we can write it as the following tuple:
\[
(U_{1,2}^{(1)}\ket{00}, {U_{3,4}^{(1)}}  \ket{00}, {U_{5,6}^{(1)}} \ket{00},  {U_{7,8}^{(1)}}\ket{00})
\]
where each entry describes a two-qubit state.

What have we lost so far?
We have lost our standard of tracking individual qubits.
Before executing the first layer, we can track each qubit individually. 
However, after applying \( U^{(1)} \), tracking single qubits results in a loss of information and precision.
Consequently, we now track pairs of qubits instead.
Despite this change, our description remains precise, and we have not lost any information about the program thus far.

The next layer of unitaries is $U^{(2)}$, which, unfortunately, destroys all tensor-product structures and can entangle each qubit with every other qubit.
\begin{align*}
\ket{\Phi_2}:=U^{(2)}\ket{\Phi_1}=U_{2,3}^{(2)}\otimes {U_{4,5}^{(2)}}  \otimes{U_{6,7}^{(2)}} \ket{\Phi_1}.
\end{align*}
This situation undermines our efforts to represent a multi-qubit state by using only tuples of the states of small numbers of qubits. It seems like the only possibility is to represent the whole state as a giant vector.

To overcome this obstacle, we will provide an alternative interpretation of what we have obtained for $\ket{\Phi_0}$ and $\ket{\Phi_1}$.
For example, consider the representation
of $\ket{\Phi_1}$:
\[
(U_{1,2}^{(1)}\ket{00}, {U_{3,4}^{(1)}}  \ket{00}, {U_{5,6}^{(1)}} \ket{00},  {U_{7,8}^{(1)}}\ket{00}).
\]
Here, each entry, e.g., $U_{1,2}^{(1)}\ket{00}$, is the joint state on qubit $1$ and qubit $2$. We can also consider it as a property of qubit $1$ and qubit $2$. One can always regard $U_{1,2}^{(1)}\ket{00}$, and other entries as well, as a linear subspace that is spanned by the vector $U_{1,2}^{(1)}\ket{00}$. Then, instead of interpreting it as

\begin{mdframed}
  The state of qubits 1 and 2 is exactly $U_{1,2}^{(1)}\ket{00}$.
\end{mdframed}

\noindent
we regard it as 

\begin{mdframed}
  The state of qubits 1 and 2
  lies in the linear subspace
  spanned by $U_{1,2}^{(1)}\ket{00}$.
\end{mdframed}

This interpretation has a long history:
ever since the inception of quantum mechanics, researchers have utilized subspaces to describe fundamental phenomena. Notably, the seminal paper \emph{The Logic of Quantum Mechanics} \cite{BvN36} employed linear subspaces as the atomic propositions of quantum logic. This approach aligns well with the framework of quantum logic, as demonstrated in that paper.

There is a subtlety here to extend it to general quantum states because qubit 1 and qubit 2 could be entangled with the rest of the qubits, i.e., the state is not in the following form: $\ket{\psi_{1,2}} \otimes \ket{\psi_{3,\cdots,8}}$.
Fortunately, this situation will not stop us from talking about the state of qubit 1 and qubit 2.
There is a concept in quantum information called the \emph{reduced density matrix} that allows us to understand subsystems within a larger quantum system.
The reduced density matrix can always be written as a positive semi-definite  matrix;
see \Cref{Se:ReducedDensityMatrices} for more details.
(Its classical counterpart is the marginal distribution of a probability distribution.)

We can now define the notion of a state satisfying an assertion, where the state is represented by a reduced density matrix and the assertion is a linear subspace.
In our example, all of the column vectors of reduced density matrix $\rho_{1,2} = U_{1,2}^{(1)}\ket{00}\bra{00}(U_{1,2}^{(1)})^{\dag}$ lie in the subspace $S_{1,2}$ of qubits $1,2$.
Operationally, we can write it as $P_{1,2}\rho_{1,2} = \rho_{1,2}$, where $P_{1,2}$ denotes the operator that projects every vector to the subspace $S_{1,2}$.
Alternatively, abusing notation to not distinguish between operator $P_{1,2}$ and subspace $S_{1,2}$, we have
\begin{equation}
  \label{Eq:SupportConstraint}
  \supp(\rho_{1,2})\subseteq P_{1,2},
\end{equation}
where $\supp(\rho)$ denotes the subspace spanned by the column vectors.
\Cref{Eq:SupportConstraint} is equivalent to 
\begin{align*}
  \ket{\Phi_1}\in P_{1,2}\otimes I_{3,4,5,6,7,8}.
\end{align*}

\noindent
The above interpretation does not lose any information about the effect of applying unitary $U^{(1)}$ to $\ket{\Phi_0} = \ket{0^{\otimes 8}}$ because the result is a single quantum state, and thus the subspace of interest is one-dimensional. 
More precisely, we know that
\[
  \{\lambda\ket{\Phi_1}: \lambda\in \mathbb{C}\}
    = P_{1,2}\otimes I_{3,4,5,6,7,8} \cap P_{3,4}\otimes I_{1,2,5,6,7,8} \cap P_{5,6}\otimes I_{1,2,3,4,7} 
           \cap P_{7,8}\otimes I_{1,2,3,4,5,6}.
\]
In other words, the intersection precisely characterizes the output state $\ket{\Phi_1}$ up to a global scalar. The scalar is usually not important because we only care about normalized vectors.

For simplicity, we can omit the identity matrices in the intersection when there is no misunderstanding. 
Now we have obtained a representation of the output state $\ket{\Phi_1}$:
\[
  (P_{1,2}^{(1)}, {P_{3,4}^{(1)}}, {P_{5,6}^{(1)}}, {P_{7,8}^{(1)}}).
\]
Unlike the previous analysis,
each projection represents
a constraint on the output state, and the tuple represents their conjunction.
These constraints uniquely identify the state $\ket{\Phi_1}$.

Next up is to apply the unitary $U^{(2)}$ to $\ket{\Phi_1}$ to obtain $U_{2,3}^{(2)}\otimes {U_{4,5}^{(2)}}  \otimes{U_{6,7}^{(2)}}\ket{\Phi_1}$.
Instead of obtaining a precise representation of the reduced density matrices, we try to get a new class of efficient constraints to identify the output state uniquely.

The conjunctive structure allows us to treat each projection individually. For instance,
\[
\begin{array}{@{\hspace{0ex}}r@{\hspace{0.75ex}}l@{\hspace{0ex}}}
                      & \ket{\Phi_1}\in P_{1,2}\otimes I_{3,4,5,6,7,8}\\ 
  \Longleftrightarrow & \, U^{(2)}\ket{\Phi_1} \in U^{(2)} [P_{1,2}\otimes I_{3,4,5,6,7,8}]{U^{(2)}}^{\dag}\\
  \Longleftrightarrow & \, U^{(2)}\ket{\Phi_1} \in (U_{2,3}^{(2)}\otimes {U_{4,5}^{(2)}}  \otimes{U_{6,7}^{(2)}}) [P_{1,2}\otimes I_{3,4,5,6,7,8}]({U_{2,3}^{(2)}\otimes {U_{4,5}^{(2)}}  \otimes{U_{6,7}^{(2)}} })^{\dag}\\
  \Longleftrightarrow &\ket{\Phi_2} \in U_{2,3}^{(2)} [P_{1,2}\otimes I_{3}] {U_{2,3}^{(2)}}^{\dag}
\end{array}
\]
Here we use the fact that for a unitary matrix $U$, $U I U^{\dag}=I$, and we regard $U_{2,3}^{(2)}$ as $I_1\otimes U_{2,3}^{(2)}$, a matrix on qubits 1,2,3. 
We define the projection $P_{1,2,3}^{(2)}$ by $P_{1,2,3}^{(2)} := U_{2,3}^{(2)} [P_{1,2}\otimes I_{3}] {U_{2,3}^{(2)}}^{\dag}$.
This projection applies nontrivially only to qubits 1, 2, and 3. 

Similarly, we can compute four other projections, and together, one has the tuple
\[
(P_{1,2,3}^{(2)},P_{2,3,4,5}^{(2)},P_{4,5,6,7}^{(2)},P_{6,7,8}^{(2)}).
\]
Here, $P_{2,3,4,5}^{(2)}$ is obtained by applying $U^{(2)}$ to $P_{3,4}^{(1)}$ and we have to work with $U_{2,3}^{(2)}\otimes {U_{4,5}^{(2)}}$ because $U_{2,3}^{(2)}$ and ${U_{4,5}^{(2)}}$ each have a non-trivial overlap with qubit 3 or 4 of $P_{3,4}^{(1)}$.
Other terms are similar.

According to our construction, $\ket{\Phi_2}$ satisfies all of these constraints, and thus satisfies their conjunction.
Now the question is, how tight is the conjunction of these constraints?
In particular, do they uniquely identify $\ket{\Phi_2}$?

The answer is ``yes,'' which is established by the following argument:
\[
\begin{array}{@{\hspace{0ex}}r@{\hspace{0.75ex}}c@{\hspace{0.75ex}}l@{\hspace{0ex}}}
  \{\lambda\ket{\Phi_1}: \lambda\in \mathbb{C}\} & = & P_{1,2}\cap P_{3,4}\cap P_{5,6} \cap P_{7,8}\\
  \Longleftrightarrow \{\lambda U^{(2)}\ket{\Phi_1}: \lambda\in \mathbb{C}\} & = & U^{(2)}[P_{1,2}\cap P_{3,4}\cap P_{5,6} \cap P_{7,8}]{U^{(2)}}^{\dag}\\
           & = & U^{(2)}P_{1,2}{U^{(2)}}^{\dag}\cap U^{(2)}P_{3,4}{U^{(2)}}^{\dag}\cap U^{(2)}P_{5,6}{U^{(2)}}^{\dag} \cap U^{(2)} P_{7,8}{U^{(2)}}^{\dag}\\
  \Longleftrightarrow  \{\lambda \ket{\Phi_2}: \lambda\in \mathbb{C}\} & = & P_{1,2,3}^{(2)}\cap P_{2,3,4,5}^{(2)} \cap P_{4,5,6,7}^{(2)}\cap P_{6,7,8}^{(2)}
\end{array}
\]
where we have omitted the identity matrices in the computation for the sake of clarity.
 
Let us now count the resources for computing this representation.
From the presentation $(P_{1,2}^{(1)}, {P_{3,4}^{(1)}}, {P_{5,6}^{(1)}}, {P_{7,8}^{(1)}} )$, we need to apply the unitary to each entry to obtain a new constraint.
On $P_{1,2}^{(1)}$ and ${P_{7,8}^{(1)}}$, this process involves two steps of $8\times 8$ matrix multiplication;
the other two---${P_{3,4}^{(1)}}$, and ${P_{5,6}^{(1)}}$---each need two steps of $16\times 16$ matrix multiplication.
The newly obtained constraints can be stored efficiently in a classical computer:
they consist of two matrices of size $8\times 8$ and two matrices of size $16\times 16$.
In summary, we only need to perform at most 8 steps of $16\times 16$ matrix multiplication to obtain this presentation, and the space cost is at most four matrices of size $16\times 16$.

The third layer is
$U^{(3)} = U_{1,2}^{(3)}\otimes {U_{3,4}^{(3)}}  \otimes{U_{5,6}^{(3)}} \otimes {U_{7,8}^{(3)}}$.
We apply a similar approach and obtain the constraints $(P_{1,2,3,4}^{(3)},P_{1,2,3,4,5,6}^{(3)},P_{3,4,5,6,7,8}^{(3)},P_{5,6,7,8}^{(3)})$ We can also confirm that the conjunction of these constraints uniquely determines the state $\ket{\Phi_3}:=U^{(3)}\ket{\Phi_2}$.

This method can be applied to quantum circuits with more layers (still constant) and more qubits. On the other hand, an efficient simulation method might not exist then \cite{terhal2004adaptivequantumcomputationconstant}. 

One might be tempted to question the effectiveness of this approach because the space required to store local projections grows, in general, exponentially in the number of qubits.
For this example, the representation of the constraints (for layer 3) consists of two square matrices of size $2^4\times 2^4$ and two matrices of size $2^6\times 2^6$.
The total space requirement is
\[
  2\times 2^4\times 2^4+2\times 2^6\times 2^6 = 8\text{,}704~\text{complex numbers}.
\]
In contrast, the output state is an eight-qubit state;
i.e., it can be represented as a $2^{8}=256$-dimensional vector of complex numbers.

Before concluding that the approach sketched above has no merit, consider what happens with a larger number of qubits, e.g., 100 qubits (but still for a depth-3 circuit with the kind of structure shown in \Cref{fig:circuit_1}).
In this case, our approach will generate a set of constraints with at most $50$ matrices, each of which is at most $2^6 \times 2^6$, and thus the total space requirement is at most $50\times 2^6\times 2^6$ $= 204\text{,}800~\text{complex numbers}$.
In contrast, a direct representation of the output state would be a complex vector of dimension $2^{100} \geq 1.26 \times 10^{30}$, which is larger by a factor of more than $10^{23}$.



\paragraph{Constraint-based descriptions and ``light cones''}
In quantum information science, light cones describe the causal structure of information propagation in quantum circuits, analogous to the concept of light cones in relativity.
One can consider both forward light cones and backward light cones.
For instance, for our constraint-based description of a constant-depth quantum circuit, a backward light cone consists of the set of qubits and gates that influence a given local projection (e.g., $P_{1,2,3,4,5,6}^{(3)}$ is influenced by qubits $\{ 1,2,3,4,5,6 \}$).
Conversely, a forward light cone consists of a set of local projections that are influenced by a given input qubit or gate (e.g., qubit 1 influences
$P_{1,2}^{(1)}$, $P_{1,2,3}^{(2)}$, $P_{1,2,3,4}^{(3)}$, and $P_{1,2,3,4,5,6}^{(3)}$).

In a constant-depth quantum circuit, the influence of any gate is restricted to a limited set of qubits within a bounded region, defining a causal light cone.
In the Noisy Intermediate-Scale Quantum (NISQ) era, these constraints make constant-depth circuits particularly valuable, because errors remain localized and do not propagate uncontrollably.

The reason that the constraint-based description of a shallow circuit remains small is essentially an argument about backward light cones.
The number of qubits involved in the local projections of a circuit's constraint-based description at layer $i+1$ grows when a two-qubit unitary links two smaller backward light cones---i.e., the backward light cone at layer $i+1$ is wider than the backward light cones at layer $i$.
In \Cref{fig:circuit_1}, we go from
$(P_{1,2}^{(1)}, {P_{3,4}^{(1)}}, {P_{5,6}^{(1)}}, {P_{7,8}^{(1)}})$ to
$(P_{1,2,3}^{(2)},P_{2,3,4,5}^{(2)},P_{4,5,6,7}^{(2)},P_{6,7,8}^{(2)})$ to
$(P_{1,2,3,4}^{(3)},P_{1,2,3,4,5,6}^{(3)},P_{3,4,5,6,7,8}^{(3)},P_{5,6,7,8}^{(3)})$;
thus, for example, the inputs to the backward light cone for $P_{1,2,3,4,5,6}^{(3)}$ are qubits $\{ 1,2,3,4,5,6 \}$.

Such enlargements depend on the geometry of the qubit arrangement.
In a 1D qubit arrangement, where two-qubit unitaries are applied only to neighboring qubits, the number of qubits in each local projection grows at most linearly in the circuit depth---i.e., from 1 to 2, then at most 4, 6, ..., up to at most \( 2d \), where \( d \) is the circuit depth. 
In contrast, in a 2D qubit arrangement, the number of qubits in the local projections grows at most quadratically in the circuit depth.
Consequently, the computation time for obtaining each local projection grows as $2^{O(d)}$ and $2^{O(d^2)}$, respectively.\footnote{
  These quantities are what were referred to in \Cref{Se:Introduction} as the qubit-interaction cost $f(d)$ in the time-complexity expression $T(n,d) = f(d) \cdot O(n)$.
}
Because there are at most $n$ local projections in our constraint-based description of a shallow circuit, the total time complexity is $2^{O(d)} \cdot O(n)$ and $2^{O(d^2)} \cdot O(n)$, respectively.
For a fixed value of $d$, both $2^{O(d)}$ and $2^{O(d^2)}$ are constants;
consequently, the cost grows linearly in the number of qubits $n$.

In contrast, if one attempts to obtain a direct representation of the output state, computation on vectors with an exponential number of dimensions quickly becomes infeasible.

\paragraph{{Application to equivalence checking}}
We now sketch how the constraint-based representation is used to perform
equivalence-checking.
For example, consider two shallow 100-qubit circuits $U_1$ and $U_2$, both of depth 3, with structures similar to \Cref{fig:circuit_1}.
Suppose that we want to check
their equivalence when 
the input state \( \ket{0}^{\otimes 100} \), i.e., $U_1\ket{0}^{\otimes 100} = \alpha U_2\ket{0}^{\otimes 100}$ for some \( |\alpha| = 1 \).
To do so, we consider the depth-6 circuit \( U_1^{\dag} U_2 \) and analyze whether $U_1^{\dag} U_2 \ket{0}^{\otimes 100} = \alpha \ket{0}^{\otimes 100}$.

To employ the local-projection-based description of the output of the quantum circuit $U_1^{\dag} U_2$, we construct a tuple of local projections, each acting on at most 12 qubits:
$(P_{s_1}^{(6)}, P_{s_2}^{(6)}, \dots, P_{s_{100}}^{(6)})$.
To verify equivalence, we check whether \( \ket{0}^{\otimes 100} \) satisfies this assertion, which reduces to ensuring that each \( P_{s_i}^{(6)} \) contains \( \ket{0}^{\otimes |s_i|} \).
The general principle used here is as follows:
\noindent
\begin{mdframed}
  Constant-depth circuits $U_1$ and $U_2$ are equivalent if and only if the ``double-depth'' (but still constant-depth) circuit $U_1^{\dag} U_2$ maps the $\ket{0}^{\otimes n}$ state perfectly to the $\ket{0}^{\otimes n}$ state.
  This property can be checked via our constraint-based description because $U_1^{\dag} U_2 \ket{0}^{\otimes n} = \alpha\ket{0}^{\otimes n}$ for some $|\alpha|=1$ if and only if $\ket{0}^{\otimes n}$ satisfies all the constraints in the constraint-based description of the circuit for $U_1^{\dag} U_2$.\footnote{
    By design, the intersection in the constraint-based description of any circuit is one-dimensional.
  }
\end{mdframed}
\noindent
In our example, the space comparison between the constraint-based approach versus computing a direct representation of the output state now becomes $100 \times 2^{12} \times 2^{12} \ll 2^{100}$.

\section{Additional Terminology and Notation}
\label{Se:AdditionalTerminologyAndNotation}

\subsection{Reduced Density Matrices and Partial Traces}
\label{Se:ReducedDensityMatrices}


We use \(\op{\psi}{\psi}\) to denote the \emph{density matrix} of a pure quantum state $\ket{\psi}$.
In contrast with a pure quantum state,
a mixed state describes a statistical ensemble of different possible pure states.
A mixed state can be represented by a density matrix \(\rho\) given by \(\rho = \sum_i p_i |\psi_i\rangle \langle \psi_i|\), where \(p_i\) are probabilities summing to 1, reflecting the lack of complete knowledge about the system's state.

A subsystem's state within a larger system from $\mathcal{H}_A \otimes \mathcal{H}_B$ can be represented by a \emph{reduced density matrix}.
The reduced density matrix for subsystem \( A \) is denoted by \( \rho_A := \Tr_{B}(\rho_{AB}):=\sum_i \langle i|_B \rho |i\rangle_B\) \cite{NI11}.
The reduced density matrix encodes the statistical properties and correlations of the subsystem, providing insight into its behavior without needing full knowledge of the entire system.
This definition can be directly generalized to the multipartite setting.

When a unitary operator \( U \) is applied to a density matrix \( \rho \), the resulting 
density matrix
is \( U \rho U^\dagger \), which is clear for \( \rho = \ket{\psi}\bra{\psi} \). When \( U \) acts on system \( A \) of a density matrix \( \rho_{AB} \), the reduced density matrix on \( A \) becomes \( U \rho_A U^\dagger \), while the reduced density matrix on \( B \) remains unchanged.
We define
\( \mathcal{U} := \lambda x.U x U^\dagger \) to denote the quantum operation on reduced density matrices that corresponds to the unitary matrix \( U \) (which operates on quantum states).

\subsection{Projections and Tuples of Projections}
\label{Se:ProjectionsAndTuplesOfProjections}

Since the inception of quantum mechanics, researchers have utilized projections to describe fundamental phenomena. The seminal paper \emph{The Logic of Quantum Mechanics} \citeN{BvN36} employed orthogonal projections as atomic propositions of quantum logic, aligning well with its framework.

\subsubsection{Projections and Projective Measurements}

\textbf{Projections:} A projective operator (projection) \( P \) satisfies \( P^2 = P=P^{\dag} \).
For example, $P=\op{00}{00}+\op{11}{11}$ is a projective operator on a 2-qubit system.
Each projective operator corresponds to a linear subspace $S_P = \{\ket{\psi}| \ P\ket{\psi}=\ket{\psi}\}$.
The correspondence between projections and subspaces facilitates a natural partial order on the set of projections. Specifically, for projections \( P \) and \( Q \), we define: $P \subseteq Q $ {if and only if} $S_P \subseteq S_Q$.
For two projections \( P \) and \( Q \) of the same dimension, \( P \subseteq Q \) if and only if \( QP = P \).

Henceforth, in a slight abuse of notation, we will use a symbol like $P$ to denote both a projection and its corresponding linear subspace $S_P$.
Which meaning is intended should be clear from context.

\textbf{Projective Measurements:} 
In quantum mechanics, measurements determine the state of a quantum system and are described by measurement operators.
Projective measurements use a set of projection operators \( \{ P_i \} \), which satisfy
$P_i^2 =  P_i$ and $ \sum_i P_i = I$ with \( I \) being the identity operator.
Such a set describes a measurement with $|\{ P_i \}|$ possible outcomes.
The probability of outcome \( i \) for a measurement on state \( \ket{\psi} \), and the post-measurement state $\ket{\psi_i'}$ on outcome $i$ satisfy
\[
  \text{Pr}(i) = \bra{\psi} P_i \ket{\psi}
  \qquad\qquad
  \ket{\psi_i'} = \frac{P_i \ket{\psi}}{\sqrt{\bra{\psi} P_i \ket{\psi}}}.
\]

\noindent
Each projection operator $P$ corresponds to a two-outcome measurement set \( \{ P, I-P \} \).

Every projection is positive semi-definite.
For any positive semi-definite matrix $A$, 
its {\em support}, written $\supp(A)$, is the subspace spanned by the eigenvectors of $A$ that have non-zero eigenvalues.
An important property of projective measurements is that a projective measurement with respect to projection $P$ does not disturb a quantum state $\ket{\psi}$ that lies in $\supp(P)$, and will output $P$ with certainty.
If \( \ket{\psi} \not\in \supp(P) \), then the projective measurement \( \{ P, I - P \} \) will yield the outcome corresponding to \( I - P \) with nonzero probability.
Moreover, if the measurement outcome is \( P \), the post-measurement state is given by $\frac{P \ket{\psi}}{\sqrt{\bra{\psi} P \ket{\psi}}} \in \supp(P)$.

Birkhoff and von Neumann \citeN{BvN36} defined that for a density matrix $\rho$ and a projection $P$, 
\[
\rho \mbox { satisfies } P \mbox{ iff } 
\supp(\rho)\subseteq P.
\]
One can easily show that 
$\rho \mbox { satisfies } P$ iff
$P\rho = P\rho P = \rho$.
This property will be used in \Cref{sec:RunTimeAssertions} to perform assertion checking.

The following lemmas justify the method sketched in \Cref{Se:Overview}.
The first lemma
explains how a constraint on a reduced density matrix corresponds to a constraint on the global quantum state.

\begin{lem}[ \cite{YP21}]
\label{lem:extension-with-identity-matrix}
For projection $P_{[n]\setminus s}$ on qubits $[n]\setminus s$, 
$\supp(\tr_{s} A) \subseteq P_{[n]\setminus s}$
iff $\supp(A) \subseteq P_{[n]\setminus s} \otimes I_{s}$.
\end{lem}

The next lemma justifies how our constraint-based representation of a circuit's state ``evolves'' when gate $U$ is the next gate to be applied during the interpretation of a circuit.

\begin{lem}[ \cite{YP21}]
\label{lem:intersection}
For projections $P_i$ and unitary $U$, 
$U(\cap_i P_i)U^{\dag}=\cap_i UP_i U^{\dag}$.
Here the intersection of the projection operators denotes the intersection of the corresponding subspaces.
\end{lem}
The proof follows from the definition of subspace intersection.

We say two projections $P$ and $Q$ \emph{commute} if $PQ=QP$.


\begin{lem}
\label{lem:commuting-2}
For commuting projections $P_i$,
$\Pi_i P_i$ is a projection corresponding to $\cap_i P_i$ 
\end{lem}
Again, the proof follows
from the definitions.


\subsubsection{Tuples of Projections as Assertions}
\label{Se:TuplesOfProjectionsAsAssertions}


As the number of qubits increases, the dimension of the vector representation of a state increases exponentially.
For any integer \(1 \leq m \leq 2^n\) and \(m\)-tuple \(S = (s_1, \cdots, s_m)\) with \(s_i \subseteq [n]\),
where $[n]=\{1,2,\cdots,n\}$,
we can use a tuple of projections $(P_{s_1}, \cdots, P_{s_m})$
to represent a tuple of constraints: each $P_{s_i}$ imposes a constraint on the qubits $\{ q_p \mid p \subseteq s_i  \}$ (and only on the qubits in $\{ q_p \mid p \subseteq s_i  \}$).
In tuple \( S \),
we impose no restrictions on the relationships between the \( s_i \)'s.
For example, it is permissible to have
a given set \( s \) appear multiple times in \( S \)---e.g., as $s_i = s$ and $s_j = s$, and it is also allowed for some \( s_k \) to be a subset of \( s_\ell \) for \( k \neq \ell \).

We introduce the assertion language $\mathcal{L}$, which
consists of
conjunctions of local projections.
For any tuple of projections, $\mathcal{P}=(P_{s_1}, \cdots, P_{s_m})$ and
reduced density matrix $\rho$ over $\overline{q}$, i.e., all qubits in $[n]$,
we say that $\ass(\overline{q};\mathcal{P})$ holds if
$\rho_{s_i} := \Tr_{[n] \setminus s_i}(\rho)$, the reduced density matrix with respect to
qubits $s_i$ (defined in \Cref{Se:ReducedDensityMatrices}), satisfies $P_{s_i}$ for each $i$.
By Lemma \ref{lem:extension-with-identity-matrix}, we can characterize the set of states where $\ass(\overline{q};\mathcal{P})$ holds as follows:
$\bigcap_i \left( P_{s_i}\otimes I_{[n]\setminus s_i} \right)$.
%

We are adopting the machinery introduced in \cite{YP21}, which uses tuples of local projections as assertions. However, we use it for a different purpose and have obtained different results compared to those in \cite{YP21}. For a comprehensive comparison, see \Cref{Se:ComparisonWithQuantumAbstractInterpretation}.


\subsection{Choi States}
\label{Se:ChoiState}

The Choi state is a fundamental concept in quantum information theory, used to represent quantum channels, including unitaries. 
The Choi state of a unitary operator \( U \) is constructed by applying \( U \) to one-half of a maximally entangled state. Specifically, if the maximally entangled state is $\ket{\Phi^+}: = \frac{1}{\sqrt{d}} \sum_{i=0}^{d-1} \ket{i}\ket{i}$
with \( d \) being the dimension of the system, the Choi state \( \rho_U \) is given by:
\[
\rho_U = (I \otimes U) \ket{\Phi^+}\bra{\Phi^+} (I \otimes U^\dagger).
\]

A key property of the Choi state proved in \cite{CHOI1975285} is that
\begin{lem}\label{choi}
For two unitary operators $U$ and $V$, we have
\[
\rho_U = \rho_V \Leftrightarrow (I \otimes U) \ket{\Phi^+}=e^{i\theta}(I \otimes V) \ket{\Phi^+}\Leftrightarrow U=e^{i\theta}V \Leftrightarrow  U\ket{\psi}=\alpha_{\psi} V\ket{\psi}, \ \forall \ \ket{\psi} \ \ \mathrm{and}\  \mathrm{some}\  |\alpha_{\psi}|=1.
\]
\end{lem}
\section{An Efficient Description of the Output of a Shallow Circuit}
\label{Se:AnEfficientDescription}

In this section, we will provide an algorithm to compute an efficient description of the output of a shallow circuit.
To reduce notational clutter, the theorems are stated for circuits of 2-qubit gates, but with only slight generalization of the notation, the results apply to circuits composed of 3-qubit gates, 4-qubit gates, etc.


Let us consider an $L$-layer circuit 
for an $n$-qubit computation.
Each layer can be represented as a tensor product of 2-qubit unitary matrices.
Let $U_{ij}^{(k)}$ denote the 2-qubit unitary acting on qubits \( i \) and \( j \) in circuit layer \( k \).
The final quantum state obtained after applying \( L \) layers of such matrices can be expressed as follows:
\[
|\psi_{\text{final}}\rangle = \left( \prod_{k=1}^{L} \bigotimes_{(i,j) \in \text{pairs}[k]} U_{ij}^{(k)} \right) |0^{\otimes n}\rangle.
\]
We focus on the case that $L$ is a constant that does not depend on $n$, and say that such a circuit is \emph{shallow}.
As we will show below, our results apply to the \emph{family} of circuits $\mathbb{C}_L$ with a fixed depth $L$ and a varying number of qubits $n$.

By generalizing the example from \Cref{Se:Overview}, we show the following theorem.
(For the proof, see
\Cref{sec:ProofOfTheorem}.)

\renewcommand{\algorithmicensure}{\textbf{Input:}}
\renewcommand{\algorithmicrequire}{\textbf{Output:}}
\begin{algorithm}[tb!]
	\caption{Compute an efficient classical description of a shallow circuit}
	  \label{alg:description}
	\begin{algorithmic}[1]
		\Ensure Classical description of unitaries $U_{ij}^{(k)}$ with $1\leq k\leq L$ in a depth-$L$ circuit $\left( \prod_{k=1}^{L} \bigotimes_{(i,j) \in \text{pairs}[k]} U_{ij}^{(k)} \right) $. 
		\Require Output a tuple of local projections $(P_{s_1},\cdots,P_{s_n})$ with $s_t\subseteq [n]$ and $|s_t|$ is bounded by a function of $L$. The conjunction of these local projections uniquely identify the circuit output $\prod_{k=1}^{L} \bigotimes_{(i,j) \in \text{pairs}[k]} U_{ij}^{(k)}  |0^{\otimes n}\rangle$.
  
      \For {$m\gets 1$ to $n$}
    \State    $s_m \gets \{m\}$  \ \ \ \ \ \ \ \ \ \ \    \ \ \ \ \ \ \ \ \ \   \ \ \ \ \ \ \ \ \ \ /* Initialize the domain*/
        
     \State   $P_{s_m}\gets \op{0}{0}$ \ \ \ \  \ \ \   \ \ \ \ \ \ \ \ \ \  \ \ \ \ \ \ \ \ \ \   /* Initialize the local projections*/
        \EndFor
   \For {$k\gets 1$ to $L$}    \ \ \   \ \ \ \ \ \ \ \ \  \ \ \ \  \ \ \ \ \ \ \ \ \ \ /* For each layer of unitaries */  \label{Li:LayerLoopBegin}

   \For {$t\gets 1$ to $n$}  \ \ \  \ \ \ \ \ \ \ \ \ \ \ \ \ \ \ \ \ /* Update each local projection $P_{s_t}$*/

    \State $r_t \gets \emptyset$
    
     \State $l_t \gets \emptyset$

   \For{each $U_{i,j}^{(k)}$} 
    \ \ \  \ \ \ \ \ \ \ \ \ \  \ \ \ /* Check whether $s_i$ has overlap with $\{i,j\}$*/

   \If{$s_t\cap \{i,j\}\neq \emptyset$}
   
   \State  $r_t \gets r_t\cup \{i,j\}$
   
   \State $l_t \gets l_t\cup \{(i,j)\}$
   \EndIf
   \EndFor
   \State $g_t\gets s_t$
   
 \State  $s_t\gets s_t\cup r_t$  \ \ \  \ \ \ \ \ \ \ \ \ \   \ \ \ \ \ \ \ \ \ \   \ \ \  \ \ \ /* Update $s_i$*/

  \State $U_t^{(k)}\gets \bigotimes_{(i,j)\in l_t} U_{i,j}^{(k)}$
   
 \State  $P_{s_t}\gets U_t^{(k)}(P_{g_t}\otimes I_{s_t\setminus g_t}){U_t^{(k)}}^{\dag}$   \ \ \  \ \ \ \ \ \ \   /* Update $P_{s_i}$*/

  \EndFor
   \EndFor        \label{Li:LayerLoopEnd}
	\end{algorithmic}
\end{algorithm}

\begin{thm}\label{main}
For a shallow circuit output $|\psi_{\text{final}}\rangle = \left( \prod_{k=1}^{L} \bigotimes_{(i,j) \in \text{pairs}[k]} U_{ij}^{(k)} \right) |0^{\otimes n}\rangle$ with constant $L$, Algorithm \ref{alg:description}
outputs a tuple of local projections $(P_{s_1},\cdots,P_{s_m})$ that uniquely identifies $|\psi_{\text{final}}\rangle $ in the following sense:
\begin{align*}
\cap_{t=1}^n (P_{s_t} \otimes I_{[n]\setminus s_t})=\{\lambda\ket{\psi_{\text{final}}}| \lambda\in \mathbb{C}\}.
\end{align*} 
Furthermore, each set $s_t$ contains at most
a constant number of
qubits, and the local projections $P_{s_1},\cdots,P_{s_m}$ pairwise commute. Specifically, the
body of the loop in lines \ref{Li:LayerLoopBegin}--\ref{Li:LayerLoopEnd}
of \Cref{alg:description} needs to be executed \(L\) times for a depth-\(L\) circuit. For a shallow circuit, where \(L\) is a constant, the time complexity is linear in \(n\) because each iteration of the loop body takes time linear in \(n\).
\end{thm}

\textbf{Remark:} The size of \( s_t \) depends on the circuit's geometry, which constrains the qubits that can interact.
In typical hardware models, such as 1D and 2D arrangements, we have \( |s_t| \leq 2d \) and \( |s_t| = O(d^2) \), respectively.
Therefore, for a fixed value of \( d \), \( |s_t| \) remains constant with respect to \( n \).
The proof in \Cref{sec:ProofOfTheorem}
does not explicitly account for the circuit's geometry; 
however, the argument holds similarly for both 1D and 2D arrangements.
\textbf{End Remark.}


\paragraph{The Bravyi-Gosset-König Algorithm}
\Cref{main} offers new insight into the Bravyi-Gosset-König Algorithm.
\citet{Bravyi_2018} showed that
there is a constant-depth circuit that can solve the following non-oracle version of the Bernstein-Vazirani problem \cite{bernstein1993quantum}, called
the 2D Hidden Linear Function problem. It is defined as follows:
given $q\, : \, \FF_2^n \to \ZZ_4:
q(x)=2\sum_{1\le \alpha< \beta\le n}  A_{\alpha,\beta}\, x_\alpha x_\beta + \sum_{\alpha=1}^n b_\alpha x_\alpha$,
where \(n=N^2\), \( x_1, \ldots, x_n, A_{\alpha,\beta}, b_\alpha \in \{0,1\} \) and \( A_{\alpha,\beta} = 1 \) if \((\alpha, \beta)\) is an edge in the \( N \times N \) grid.
The goal is to
output some $z$ such that $q(x)=2z^T x$ for all $x\in L_q$, where $L_q=\{ x\in \FF_2^n \, : \,  q(x\oplus y)=q(x) + q(y) \quad \forall y\in \FF_2^n\}$.

Let \( K_q \) denote the set of all solutions \( z \).
\citeauthor{Bravyi_2018} showed that their constant-depth quantum circuit can generate an equal superposition of
the
elements in \( K_q \), i.e., $\frac{1}{|K_q|} \sum_{z \in K_q} \ket{z}$.
As a consequence of \Cref{main}, the equal superposition of solutions in $K_q$ to the 2D Hidden Linear Function is uniquely determined by the constant-size collection of local projections.
Thus, our approach unveils a previously unknown feature of \( K_q \).

\section{Equivalence Checking}
\label{Se:EquivalenceChecking}

This section presents efficient algorithms to resolve the two equivalence-checking problems for shallow quantum circuits posed in \Cref{Se:Introduction} (see \Cref{{equi-w},{equi-s}} below).
Let us consider two circuits that define two $n$-qubit systems,
\[
  C_0: \left( \prod_{k=1}^{L} \bigotimes_{(i,j) \in \text{pairs}[k]} U_{ij}^{(k)} \right) \qquad\mathrm{and}\qquad
  C_1: \left( \prod_{k=1}^{T} \bigotimes_{(r,t) \in \text{pairs}[k]} V_{rt}^{(k)} \right).
\]
We want to check whether they are equivalent or not. Two different definitions of equivalence are of interest.
Both are motivated by the property that quantum states that differ by a phase factor \(\exp(i\theta)\) can be considered identical because they behave the same under quantum operations.

\begin{definition}[Weak equivalence]\label{equi-w}
Two circuits, \(C_0\) and \(C_1\), are equivalent if the output states obtained after executing them on \(\ket{0}^{\otimes n}\) are identical up to a phase factor.
\end{definition}

We are also interested in the following, stronger notion of equivalence:

\begin{definition}[Strong equivalence]\label{equi-s}
Two circuits, \(C_0\) and \(C_1\), are equivalent
if, for each $n$-qubit state \(\ket{\psi}\), the respective output states
obtained after executing them on \(\ket{\psi}\) are identical up to a phase factor.
\end{definition}

\subsection{Weak Equivalence Checking under \Cref{equi-w}}
\label{sec:EquivalenceCheckingOne}

Let us consider the problem of checking whether $C_0$ and $C_1$ are equivalent under \Cref{equi-w}.

An immediate idea might be to use Theorem \ref{main} to compute descriptions of 
$\left( \prod_{k=1}^{L} \bigotimes_{(i,j) \in \text{pairs}[k]} U_{ij}^{(k)} \right)
\ket{0}^{\otimes n}$ and $\left( \prod_{k=1}^{T} \bigotimes_{(r,t) \in \text{pairs}[k]} V_{rt}^{(k)} \right)\ket{0}^{\otimes n}$, denoted by $(P_{s_1},\cdots,P_{s_n})$ and $(Q_{t_1},\cdots, Q_{t_n})$, respectively. 
However, the respective tuples of qubit sets
$(s_1,\cdots,s_n)$ and $(t_1,\cdots,t_n)$ can be very different, even if $C_0$ and $C_1$ are equivalent under \Cref{equi-s}.
Consequently, this approach would not, in general, allow us to check the equivalence of
shallow circuits because there is no canonical form for $(P_{s_1}, \ldots, P_{s_n})$ in
our shallow-circuit descriptions.

To overcome this issue,
we observe the following:
\[
\begin{array}{@{\hspace{0ex}}r@{\hspace{0.75ex}}l@{\hspace{0ex}}}
                        & U\ket{0}^{\otimes n}=\alpha V\ket{0}^{\otimes n},~ ~\mathrm{for~some}~~ |\alpha|=1\\
  \Longleftrightarrow \,& \alpha^{-1}V^{\dag}U\ket{0}^{\otimes n} =\ket{0}^{\otimes n}  \\
  \Longleftrightarrow \,& \ket{0}^{\otimes n}~\mathrm{satisfies}~(R_{l_1},\cdots,R_{l_n}),
\end{array}
\]
where $U:=\left( \prod_{k=1}^{L} \bigotimes_{(i,j) \in \text{pairs}[k]} U_{ij}^{(k)} \right) $, $V:=\left( \prod_{k=1}^{T} \bigotimes_{(r,t) \in \text{pairs}[k]} V_{rt}^{(k)} \right)$, and $(R_{l_1},\cdots,R_{l_n})$ is the description of $V^{\dag}U\ket{0}^{\otimes n}$ obtained by Algorithm \ref{alg:description}.
This condition can be checked efficiently because the new circuit is still shallow:
the value of $L+T$ is independent of the number of qubits $n$.

\subsection{Equivalence Checking under \Cref{equi-s}}
\label{sec:EquivalenceCheckingTwo}

This section provides an efficient algorithm for equivalence checking as defined in \Cref{equi-s}.

\Cref{equi-w} only requires equal output states for the input $\ket{0}^{\otimes n}$, while \Cref{equi-s} requires equal outputs for every possible input state $\ket{\psi}$. It is not feasible to range over all $n$-qubit states $\ket{\psi}$, because there are infinitely many, most of which are highly entangled and require an exponential number of parameters to describe. Our idea to check for equivalent operations by checking for equivalent Choi states, thanks to \Cref{choi}, with $C_0: U = \left( \prod_{k=1}^{L} \bigotimes_{(i,j) \in \text{pairs}[k]} U_{ij}^{(k)} \right)$ and $C_1: V = \left( \prod_{k=1}^{T} \bigotimes_{(r,t) \in \text{pairs}[k]} V_{rt}^{(k)} \right)$. Omitting the normalization factor $\frac{1}{2^n}$, we only need to consider $C_0$ and $C_1$ as shallow circuits on $2n$ qubits and take a single input state $\sum_{m=0}^{2^n-1}\ket{m} \ket{m}$.
{\small
\begin{align*}
  &U \ket{\psi}=\alpha_{\psi}V\ket{\psi} \  \ \forall \ket{\psi}, ~\mathrm{for~some}~ |\alpha_{\psi}|=1\\
  \Leftrightarrow &\sum_{m=0}^{2^n-1}\ket{m} U \ket{m}=\alpha\sum_{m=0}^{2^n-1}\ket{m}V\ket{m}, ~\mathrm{for~some}~ |\alpha|=1\\
  \Leftrightarrow &\left(I_{[n]}\otimes U \right) \sum_{m=0}^{2^n-1}\ket{m} \ket{m}=\alpha\left( I_{[n]}\otimes V \right)\sum_{m=0}^{2^n-1}\ket{m} \ket{m}, ~\mathrm{for~some}~ |\alpha|=1.
\end{align*}
}
The last step is to observe that $\sum_{m=0}^{2^n-1}\ket{m} \ket{m}$ can be generated with a single layer of unitaries,
\[
\frac{1}{\sqrt{2^n}}\sum_{m=0}^{2^n-1}\ket{m} \ket{m}=(\frac{1}{\sqrt{2}}(\ket{00}+\ket{11}))^{\otimes n}=\otimes_{p=1}^n \mathit{CNOT}(H\otimes I)\ket{00}).
\]
Here, we pair the original $n$ qubits with the additional $n$ qubits as $(p, p')$, and create entanglement between each pair $p$ and $p'$ without entangling distinct pairs.
One can verify that $W\ket{00}=\frac{1}{\sqrt{2}}(\ket{00}+\ket{11})$ with two-qubit gate $W := \mathit{CNOT}(H\otimes I)$.

In summary, we have reduced the problem of checking the equivalence of $C_0$ and $C_1$ under \Cref{equi-s} to the problem of checking the equivalence of $C_0'$ and $C_1'$ on $\ket{0}^{\otimes n}$, where
\[
\begin{array}{@{\hspace{0ex}}r@{\hspace{0.75ex}}c@{\hspace{0.75ex}}l@{\hspace{0ex}}}
  C_0' & := & \left(I_{[n]}\otimes \prod_{k=1}^{L} \bigotimes_{(i,j) \in \text{pairs}[k]} U_{ij}^{(k)} \right) \mathit{W}^{\otimes n},  \\
  C_1' & := & \left( I_{[n]}\otimes\prod_{k=1}^{T} \bigotimes_{(r,t) \in \text{pairs}[k]} V_{rt}^{(k)} \right) \mathit{W}^{\otimes n}.
\end{array}
\]
 This problem is solved in
 \Cref{sec:EquivalenceCheckingOne}.

Note
that inserting the $\mathit{W}$ gates increases the circuit depth by at most 1 because $W$
is a 2-qubit gate.
Consequently, the application of $W^{\otimes n}$ just augments the circuit with a single layer.
 
\section{Efficient Runtime Assertions for Testing and
Debugging}
\label{Se:AssertionChecking}

In this section, we employ the efficient description of the output of a shallow circuit from \Cref{Se:AnEfficientDescription} as an efficient scheme for assertion checking, static and run time, respectively.


\subsection{Efficient Verification of Assertions}\label{veri}

Given a quantum circuit and a tuple of local projections, can we efficiently verify whether the circuit's outcome $|\psi_{\text{final}}\rangle = \left( \prod_{k=1}^{L} \bigotimes_{(i,j) \in \text{pairs}[k]} U_{ij}^{(k)} \right) |0^{\otimes n}\rangle$ satisfies the conjunction of the local-projection assertions $(Q_{l_1},\cdots,Q_{l_m})$?
That is,
can we efficiently test the property
``$|\psi_{\text{final}}\rangle ~\mathrm{satisfies}~ Q_{l_i}\otimes I_{[n]\setminus l_i}, ~\mathrm{for~all}~ i$''?
We answer this question affirmatively for shallow circuits.

Before doing so, let us make an initial attempt: suppose that we employ the outputs $(P_{s_1},\cdots,P_{s_m})$ of Algorithm \ref{alg:description}, which uniquely identifies $|\psi_{\text{final}}\rangle $.
We then need to check whether
\[
\bigcap_{j}(P_{s_j}\otimes I_{[n]\setminus s_j})\subseteq Q_{l_i}\otimes I_{[n]\setminus l_i} \ \ \forall \ \ i.
\]
Unfortunately, it is not clear how to check this condition efficiently with a classical computer.

{
The
}
solution is to work \emph{backwards}.
That is, instead of working forwards using Algorithm \ref{alg:description} to compute a description of $|\psi_{\text{final}}\rangle$, we
work backwards to apply the inverse of the circuit on the given assertion $(Q_{l_1},\cdots,Q_{l_m})$ and check whether the initial state $\ket{0}^{\otimes n}$ satisfies a set of derived local-projection assertions. According to Lemma \ref{lem:intersection} and Lemma \ref{lem:extension-with-identity-matrix}, we have
\[
\begin{array}{@{\hspace{0ex}}r@{\hspace{0.75ex}}l@{\hspace{0ex}}}
  &|\psi_{\text{final}}\rangle ~\mathrm{satisfies}~ Q_{l_t}\otimes I_{[n]\setminus l_t}, \ \ \forall \ \ t \\
  \Longleftrightarrow &\ket{0}^{\otimes n} ~\mathrm{satisfies}~  \left( \prod_{k=1}^{L} \bigotimes_{(i,j) \in \text{pairs}[k]} U_{ij}^{(k)} \right) ^{\dag}(Q_{l_t}\otimes I_{[n]\setminus l_t})\left( \prod_{k=1}^{L} \bigotimes_{(i,j) \in \text{pairs}[k]} U_{ij}^{(k)} \right), \forall \ \ t\\
  \Longleftrightarrow &\ket{0}^{\otimes n} ~\mathrm{satisfies}~  Q_{v_t}\otimes I_{[n]\setminus v_t},  \forall \ \ t \\
  \Longleftrightarrow &\ket{0^{\otimes |v_t|}} ~\mathrm{satisfies}~ Q_{v_t},\forall \ \ t. 
\end{array}
\]
Here, $Q_{v_t}\otimes I_{[n]\setminus v_t}$ denotes $\left( \prod_{k=1}^{L} \bigotimes_{(i,j) \in \text{pairs}[k]} U_{ij}^{(k)} \right) ^{\dag}(Q_{l_t}\otimes I_{[n]\setminus l_t})\left( \prod_{k=1}^{L} \bigotimes_{(i,j) \in \text{pairs}[k]} U_{ij}^{(k)} \right)$.
Similar to the proof of Theorem \ref{main}, $Q_{v_t}$ can be computed efficiently from local projection $Q_{l_t}$, where each $l_t$ only involves a constant number of qubits (independent of $n$).
The checking of whether $\ket{0^{\otimes |v_t|}}$ satisfies $Q_{v_t}$ is immediate. 

\subsection{Run-Time Assertions}
\label{sec:RunTimeAssertions}

We first recall a previous work Proq \cite{li2019proq}, which defines assertions using general projections and proposes efficient implementation strategies by deriving local projections from general projections. 

\begin{definition}[\textbf{Syntax} and \textbf{Semantics} of the assertion in \cite{li2019proq}]  \label{semantics}
The \textbf{syntax}  of the quantum assertion is defined as:
$\ass(\overline{q};P)$ where $\overline{q}=q_1,...,q_n$ is a collection of quantum variables and $P$ is a projection in the state space $\cH_{\overline{q}}$. The \textbf{semantics} of an $\ass$ statement is defined as follows:
\[
\begin{array}{@{\hspace{0ex}}r@{\hspace{0.75ex}}l@{\hspace{0ex}}}
    \ass(\overline{q};P)\equiv & \ \mathbf{if}\ M_{P}[\overline{q}] =m_0\rightarrow \ \mathbf{skip} \\
                 &\ \square \hspace{1.4cm} m_1\rightarrow \ \mathbf{abort} \\
                 &\ \mathbf{fi}
\end{array}
\]
 where $M_{P} = \{M_{m_0} = P, M_{m_1} = I_{\cH_{\overline{q}}}-P\}$ and an auxiliary notation $\mathbf{abort}$ is employed to denote that the program terminates immediately and reports the termination location.
\end{definition}
That is, one can construct a projective measurement $M_{P} = \{M_{m_0} = P, M_{m_1} = I_{\mathcal{H}_{\overline{q}}} - P\}$ and apply it to the qubit collection $\overline{q}$.
If the state satisfies the predicate, the measurement outcome is \( m_0 \), and the program proceeds.  
This
mechanism relies on an
important property of projection: a successful projective measurement does not disturb a quantum state that lies within the corresponding subspace.
If the outcome is $m_1$, the state does not satisfy the predicate, and the program terminates, reporting the termination location. However, this approach \cite{li2019proq} faces scalability challenges, both in its classical and quantum components. It requires the classical computation of the exponential size projection and the execution of corresponding projective measurements via complex unitary circuits, which could result in exponential-time costs on a quantum computer.

Before presenting a general method for designing and implementing scalable assertions, we define the following semantics for commuting local projections.

\begin{definition}[\textbf{Semantics for commuting local projections as
an assertion}]
For a tuple of pairwise-commuting local projections $\mathcal{P}=(P_{s_1},\cdots,P_{s_m})$, the \textbf{semantics} of ``$\ass(\overline{q};\mathcal{P})$'' is
\[
    \ass(\overline{q};\mathcal{P}) \equiv \ass(\overline{q};P_{s_1}),\cdots, \ass(\overline{q};P_{s_m}).
\]
where the assertion $\ass(\overline{q};P_{s_1})$ for each component is defined as in \Cref{semantics}.
\end{definition}
Clearly, $\ass(\overline{q};\mathcal{P})$ implements a two-outcome measurement, where---by \Cref{lem:commuting-2}---the result $m_0$ corresponds to the operator
$
  \Pi_{i} (P_{s_i}\otimes I_{[n]\setminus s_i})
  =
  \bigcap_i P_{s_i}\otimes I_{[n]\setminus s_i}.
$
Furthermore, we can see that the order of invocation of the set of operators $P_{s_i}$ does not affect $M_{\mathcal{P}}$.

The commuting condition is necessary because one can easily construct \( P_{1,2} \) and \( P_{2,3} \) such that \( P_{1,2}P_{2,3} \neq P_{2,3}P_{1,2} \), which could introduce ambiguity in the definition.

\subsection{NISQ-Device Verification}
\label{sec:NISQDeviceVerification}

Given a NISQ device that claims to implement a shallow circuit $|\psi_{\text{final}}\rangle = \left( \prod_{k=1}^{L} \bigotimes_{(i,j) \in \text{pairs}[k]} U_{ij}^{(k)} \right) |0^{\otimes n}\rangle$, we can design a verification scheme to assess its implementation.

We first compute the efficient description 
$\mathcal{P}:=(P_{s_1},\cdots, P_{s_m})$
developed in \Cref{Se:AnEfficientDescription}. 
The key feature is that for shallow circuits, the description
$(P_{s_1},\cdots, P_{s_m})$
satisfies $P_{s_i}\otimes I_{[n]\setminus s_i}$ with pairwise-commuting
projections $P_{s_i}$ (see \Cref{main}).
When we use the above semantics for computing projections, and obtain the description 
$(P_{s_1},\cdots, P_{s_m})$,
each correct run of the hardware implementation of the circuit will produce a quantum state $|\psi_{\text{final}}\rangle$ for which the following holds:
\begin{equation}
  \label{Eq:PropertyOfCorrectRun}
  \bigcap_i P_{s_i}\otimes I_{[n]\setminus s_i}=\op{\psi_{\text{final}}}{\psi_{\text{final}}}.
\end{equation}
In other words, the test $M_{P}[\overline{q}] = m_0$ in the code given in \Cref{semantics} can be performed by checking whether \Cref{Eq:PropertyOfCorrectRun} holds---which can be directly checked via a small number of local-projective measurements, rather than a giant projection, without loss of any precision except those caused by statistical fluctuations.

With the local-projection description, we can insert the assertion $\ass(\overline{q};\mathcal{P})$ at the end of the program. If an error message occurs in $\ass(\overline{q};\mathcal{P})$,
we conclude that the implementation is not correct. If no error message is reported after executing the program for many times, we claim that implementation is close to bug/noise-free.

%

Recall that for the example from \Cref{Se:Overview}, we obtain the following 
local-projection description
$\mathcal{P}$ $= (P_{1,2,3,4}^{(3)},P_{1,2,3,4,5,6}^{(3)},P_{3,4,5,6,7,8}^{(3)},P_{5,6,7,8}^{(3)})$.
To verify the correctness of a run of an implementation of the computation in hardware---i.e., to check $\ass(\overline{q};\mathcal{P})$---we
only need to perform four measurements, each involving at most six qubits.
For a 100-qubit variant of the circuit from \Cref{Se:Overview}, we would have to perform fifty measurements, but each would still involve at most six qubits.
\vspace{-3mm}
\section{Evaluation}
\label{sec:Experiments}

The goal of our evaluation was to answer the following questions:

\begin{description}
  \item[RQ1] How effective is the constraint-based approach to computing (a characterization of) a shallow quantum circuit's output state, compared with state-vector simulation?
  \item[RQ2] How effective is the constraint-based approach for equivalence  
  and inequivalence checking,
  compared with state-vector simulation?
  To what degree does floating-point error arise?
\end{description}

\textbf{Implementation}.
Our implementation consists of about 2,000 lines of Python, using the Qiskit library (Version 1.2.4).
The Qiskit interface was used to handle quantum-circuit descriptions and state-vector simulation.
For the local-projection-based method, matrix multiplication is performed using the standard matrix operations provided by NumPy (Version 2.1.1.).

\textbf{Benchmarks and benchmarking.}
All benchmarks used are random circuits, where every 2-qubit gate is generated from a Haar distribution, created for a specified number of qubits $n$ and a specified depth $d$.
All numbers reported are the average running time for 100 random circuits (one run per circuit).
For the equivalence-checking experiments, each run used two copies of the same random circuit.
For the inequivalence-checking experiments, each run used two different random circuits.
(It would be highly unlikely that two such circuits would perform the same computation.)

All experiments were run on a server equipped with 2× Intel\textsuperscript{\textregistered} Xeon\textsuperscript{\textregistered} Gold 6338 CPUs (128 threads total) and 1.0~TiB of RAM, running Ubuntu 20.04.6 LTS. State-vector simulation was performed using Qiskit Aer (Version~0.17.0) under noiseless conditions.
Due to memory constraints and to ensure computational stability, all experiments involving explicit state-vector simulations were limited to circuits with at most 34 qubits.

\textbf{RQ1: Efficiency of local-projection descriptions versus explicit state-vector simulation (\Cref{fig:Running_times_depth_changes}).}
In the first experiment, we fixed the depth at $d=6$ and varied the number of qubits from $10 \leq n \leq 200$ for local-projection descriptions, and from $10 \leq n \leq 34$ 
for state-vector simulation.
As shown in \Cref{fig:fixed_depth_changed_n_qubit}, the time taken using local-projection descriptions increases
roughly
linearly in $n$, whereas explicit state-vector simulation increases exponentially in $n$.

\begin{figure}[!tb]
    \centering
   []{
        \includegraphics[width=0.45\linewidth]{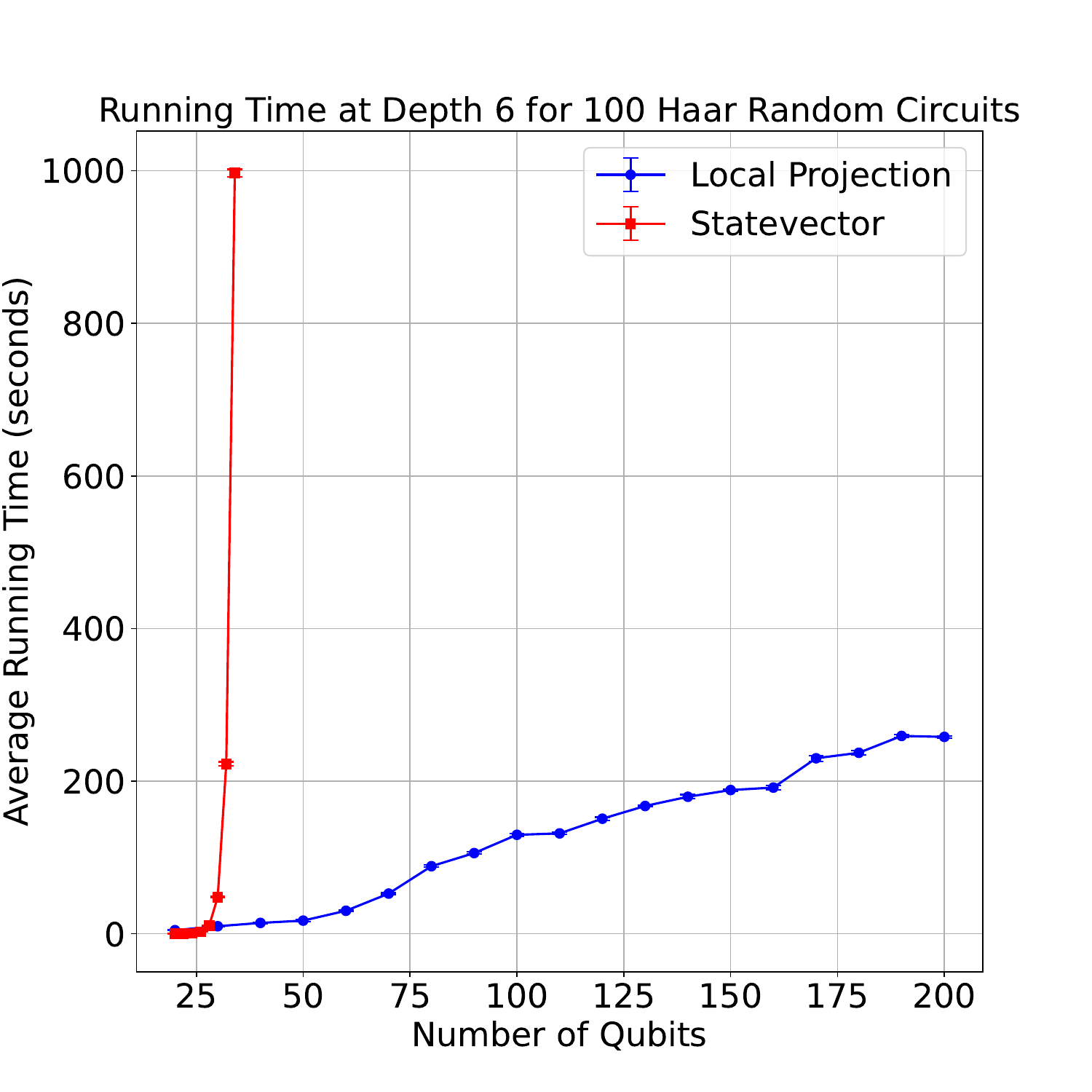}
        \label{fig:fixed_depth_changed_n_qubit}
    }\hspace{5mm}
   []{
        \includegraphics[width=0.45\linewidth]{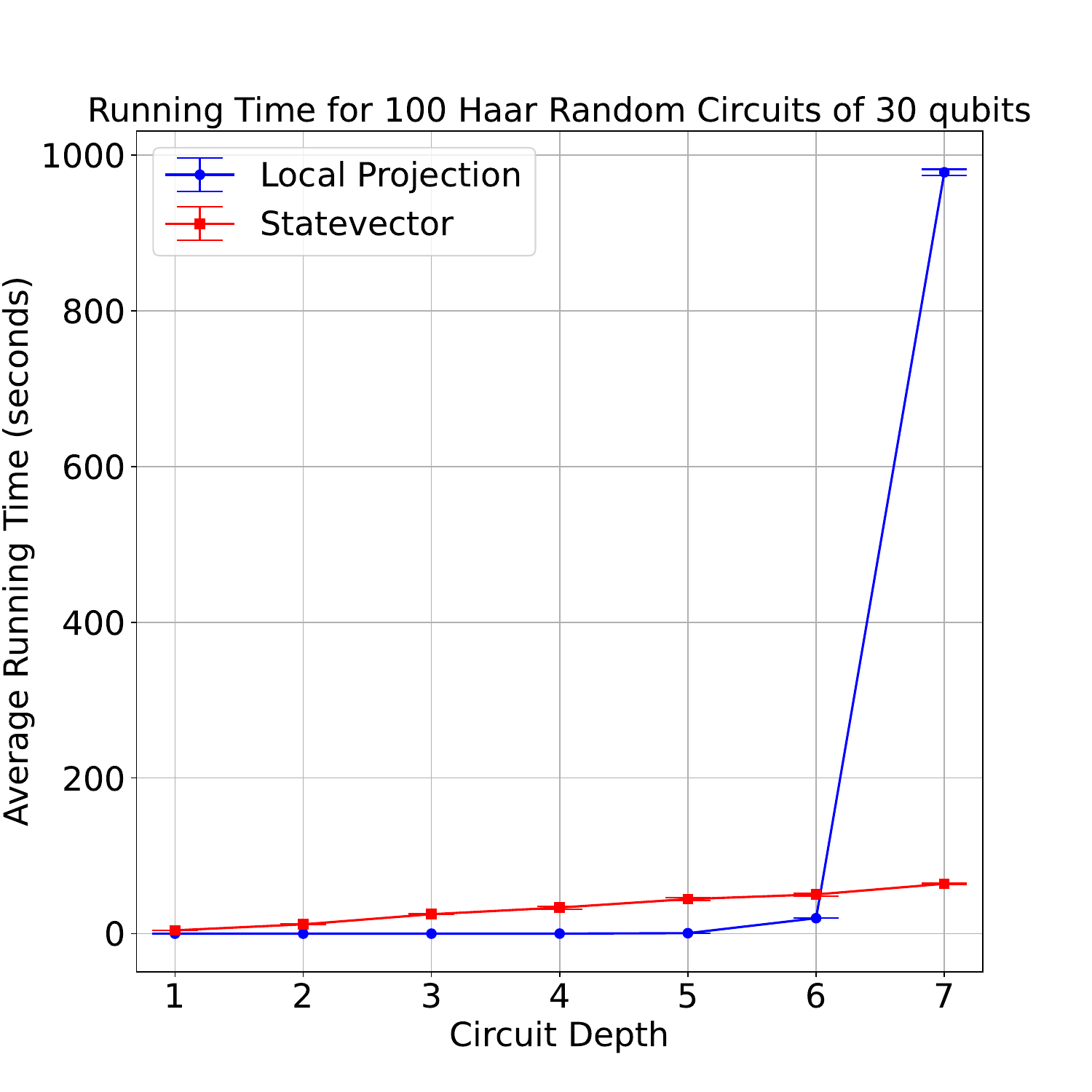}
        \label{fig:fixed_number_of_qubits}
    }
    \vspace{-1.0em}
    \caption{ Running times for computing local-projection descriptions and explicit state vectors. (a) Depth is fixed at $d=6$. (b) The number of qubits is fixed at $n=30$. 
    Each point is averaged over 100 random circuits.
    Error bars
    (visible when the figure is magnified) indicate one
    standard deviation.
   }
   \label{fig:Running_times_depth_changes}
\end{figure}

In the second experiment, we fixed the number of qubits at $n=30$, and compared the running times for computing the local-projection description and the explicit state vector as a function of circuit depth $d$.
As we would expect, \Cref{fig:fixed_number_of_qubits} shows that the time for computing explicit state vectors grows linearly in $d$.
The time for computing local-projection descriptions increases exponentially with $d$ due to the increasing size of each projection.
The results for $d=8$---not included in \Cref{fig:fixed_number_of_qubits} so as not to swamp the $y$-axis scale---were as follows: 42,523 seconds ($\approx$11.8 hours) to compute the local-projection description, and 78.1 seconds to compute the state vector.


\textit{Findings}. Our experiments confirm the expected result that for an $n$-qubit shallow circuit of depth $d$, explicit state-vector simulation scales exponentially with $n$ and linearly with $d$ due to the Hilbert-space size and gate layers.
In contrast, the local-projection computation scales linearly with $n$ and increases exponentially with $d$, because the number of required projections grows linearly with $n$, and the sizes of the projection matrices increase exponentially with $d$.
The experiments demonstrate the efficiency of local-projection descriptions as a representation of the output of a shallow quantum-circuit, suitable for use on classical computers.

\textbf{RQ2: Equivalence and inequivalence checking.}
To check equivalence of two circuits $C_1$ and $C_2$ under \Cref{equi-w}, we computed the set of local projections $\{ P^{12}_i \}$ of the composite circuit $C_1^\dagger C_2$, and checked if $\ket{0}^{\otimes n} \in \cap_i P_i^{12}$.
To check equivalence under \Cref{equi-s}, we implemented the method from \Cref{sec:EquivalenceCheckingTwo}.
%

\begin{figure}[!tb]
    \centering
    \includegraphics[width=0.8\linewidth]{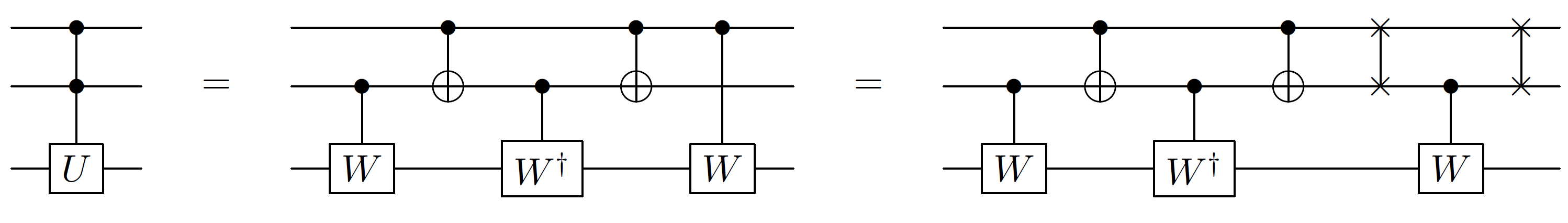}
    \vspace{-1.0em}
    \caption{Three implementations of a
    controled-controlled-$U$
    gate \cite{NI11}, where $W^2=U$.}
    \label{fig:ccu}
\end{figure}

\begin{figure}[!tb]
    \centering
    \includegraphics[width=0.5\linewidth]{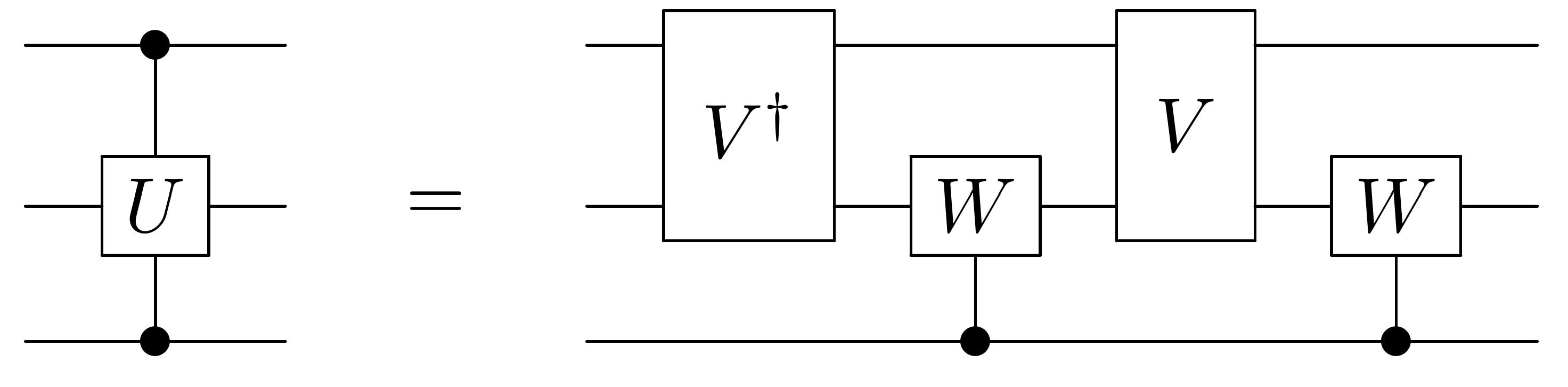}
    \vspace{-1.0em}
    \caption{For \( U = \text{diag}(e^{-i\theta}, e^{i\theta}) \), we have a depth-4 implementation of a controlled-controlled-$U$ gate \cite{PhysRevA.91.032302}, where \( W^2 = U \),
    and 2-qubit gate $V$ is defined by $V\ket{00} = \ket{01}$, $V\ket{01} = \ket{00}$, $V\ket{10} = \ket{10}$, and $V\ket{11} = \ket{11}$.
}
    \label{fig:4 Gate}
\end{figure}

\textit{Two micro-benchmarks: {equivalence checking of 3-qubit gates}.}
The two micro-benchmarks are based on the $3$-qubit circuits shown in \Cref{fig:ccu,fig:4 Gate}.
(In each figure, we refer to the left-hand-side circuit as $C_1$ and the right-hand-side circuit as $C_2$.) For circuits with depth comparable to qubit count, our equivalence-checking methods apply but provide no computational advantage over standard simulation. (In a 1D circuit, the final layer’s light cone can contain up to \( 2d \) qubits; hence, for circuits with few qubits, the projection matrices may cover all qubits.
)
To emulate a larger-scale example, we embedded $C_1^\dagger C_2$ in a 20-qubit circuit whose structure is similar to \Cref{fig:circuit_1}, in which all gates other than the ones that implement $C_1^\dagger$ and $C_2$ are 2-qubit identity gates.
Our implementation does not treat identity gates specially, so the performance would be comparable to checking the equivalence of (i) a 20-qubit circuit $C$---with gates other than identity gates---that contains $C_1$ as a sub-circuit, and (ii) a circuit $C'$ that is identical to $C$, except that $C_1$ is replaced by $C_2$. 

\textit{(1) Deutsch gate:}
The left-hand circuit $C_1$ in \Cref{fig:ccu} denotes a 3-qubit ``controlled-controlled-$U$'' ($CC\text{-}U$) gate \cite{deutsch1985quantum}.
(It is a 1D circuit when 3-qubit gates are permitted.)
The middle circuit implements $CC\text{-}U$ using only 2-qubit gates, where $W$ is a unitary that satisfies $W^2=U$.
However, this circuit is not a 1D circuit 
(when one is restricted to 2-qubit gates)
because the final occurrence of $W$ uses qubits 1 and 3.
The right-hand circuit $C_2$ is a 1D circuit that uses two swap gates to emulate the middle circuit;
its depth is 7, but the second $\textit{CNOT}$ and the adjacent swap can be merged into a single 2-qubit gate to create a depth-6 circuit.
The composite circuit $C_1^\dagger C_2$ has depth 7.

We used a Haar random matrix for $W$ (which was squared to obtain $U$), and embedded $C_1^\dagger C_2$ into a 20-qubit, depth-7 circuit in the style of \Cref{fig:circuit_1}.
Checking weak equivalence took about 27 minutes;
checking strong equivalence took about 119 minutes.
Both checks succeeded.

\textit{(2) $D(-\theta,\theta)$ gate:}
For the $CC\text{-}U$ gate with \( U = \text{diag}(e^{-i\theta}, e^{i\theta}) \), we checked the equivalence of the circuits $C_1$ and $C_2$ shown on the left-hand and right-hand sides of \Cref{fig:4 Gate}, respectively.
$C_1^\dagger C_2$ has depth 5, and $C_1$ involves a 3-qubit gate.
We again used a 20-qubit circuit in the style of \Cref{fig:circuit_1} in which we embedded $C_1^\dagger C_2$.
The running time was around 1 second for checking both weak equivalence and strong equivalence.
Both checks succeeded.


\begin{figure}[!tb]
    \centering
    
   []{
        \includegraphics[width=0.44\linewidth]{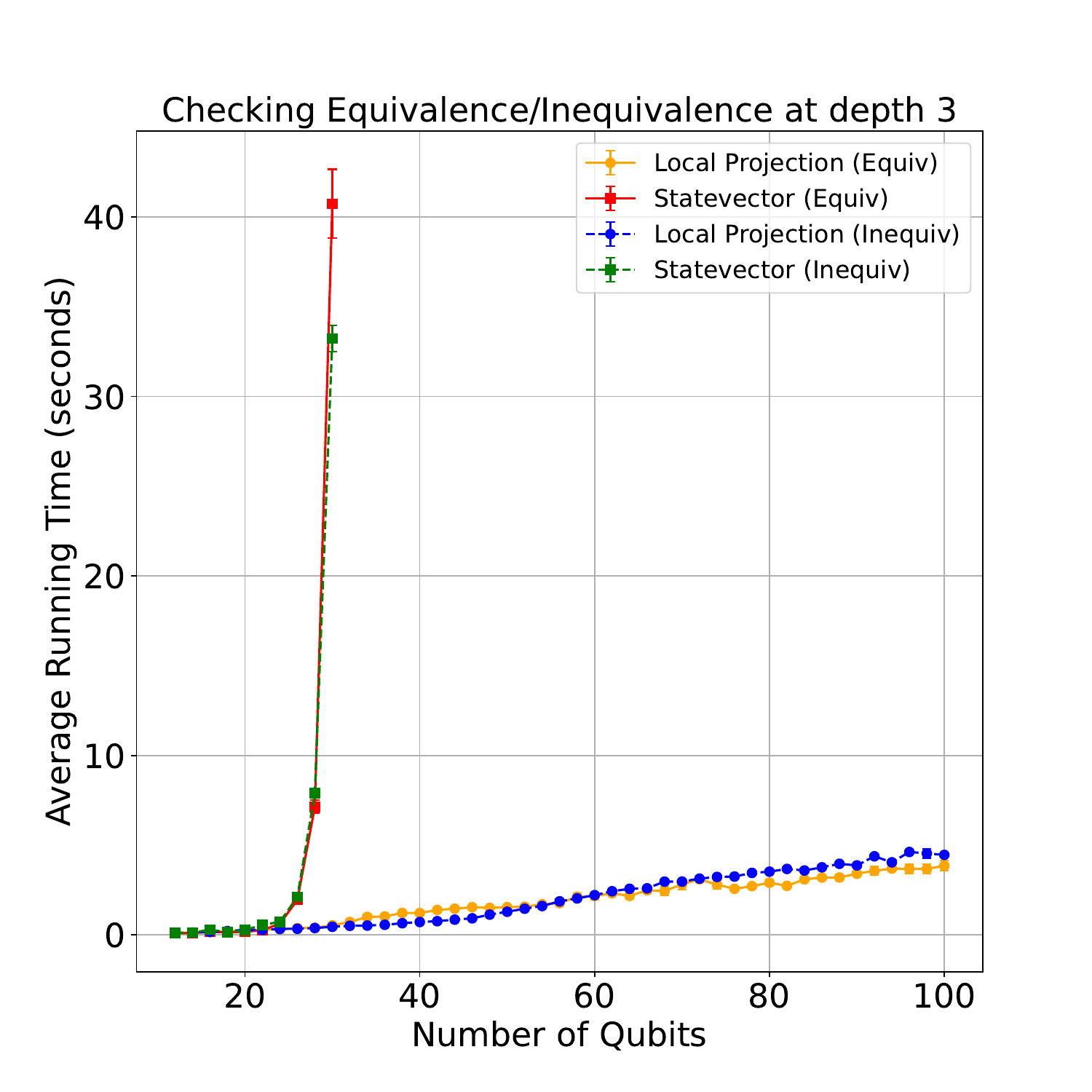}
        \label{fig:running_time_checking_equivalence}
    }\hspace{5mm}
   []{
        \includegraphics[width=0.44\linewidth]{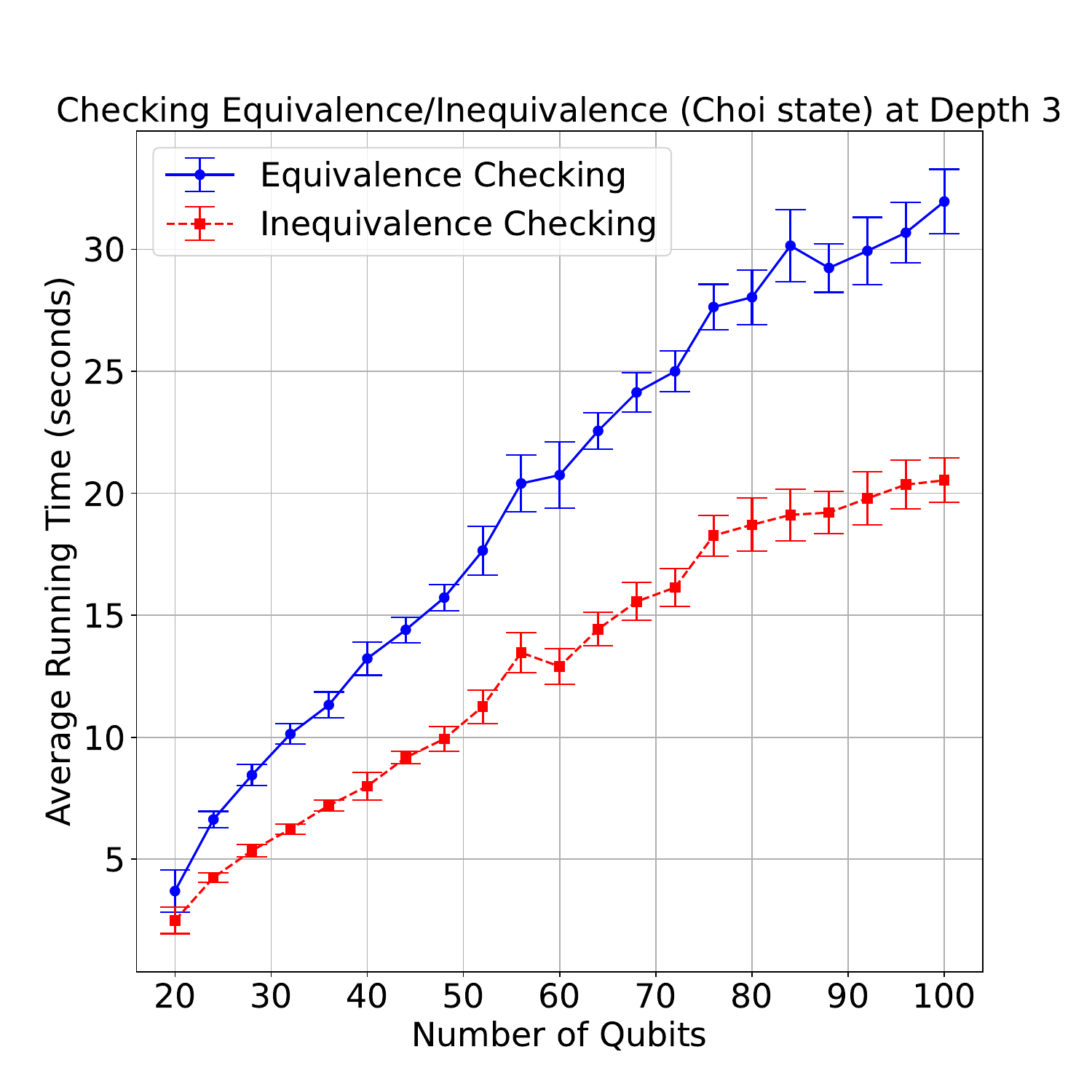}
        \label{fig: Choi approach}}
    \vspace{-1.0em}
    \caption{
    Graphs of running time as a function of number of qubits for (a) checking equivalence and inequivalence under \Cref{equi-w} (i.e., fixed input state $\ket{0}^{\otimes n}$), and (b) checking equivalence/inequivalence under \Cref{equi-s}.
    Error bars (visible in (a) when the figure is magnified) indicate one
    standard deviation.
    }
    \label{fig: equivalence checking all}
\end{figure}

\begin{figure}
    \centering
    \includegraphics[width=0.44\linewidth]{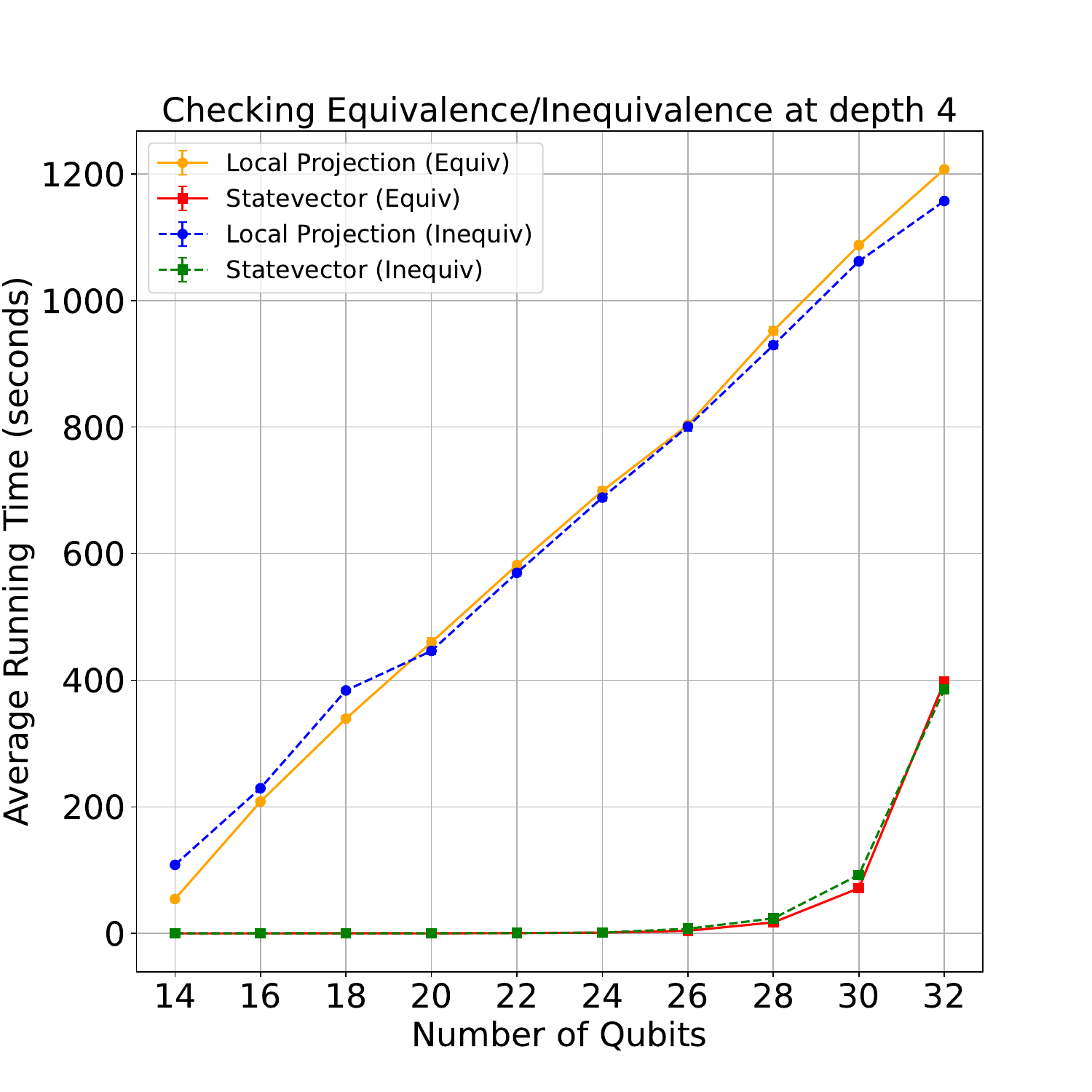}
    \caption{Checking equivalence/inequivalence of depth-4 circuits under \Cref{equi-w} (weak equivalence).
    }
    \label{fig:depth-4-euivalence checking}
\end{figure}

\textit{Equivalence and inequivalence checking on larger circuits (\Cref{fig: equivalence checking all}).}
In these experiments, we used
Haar random
circuits of depth 3, giving $U_1^\dagger U_2$ a depth of 6.
\Cref{fig:running_time_checking_equivalence} compares run times
both for runs when $U_1 = U_2$ (i.e., equivalence checking), and for runs when $U_1 \neq U_2$ (i.e., inequivalence checking).
For both kinds of problems, the local-projection-based method exhibits roughly linear growth in run time as a function of number of qubits, proving far more efficient than the state-vector approach, which has exponential growth (capped at $n=30$
in our experimental setup).




\Cref{fig: Choi approach} shows running times for checking equivalence and inequivalence in the sense of \Cref{sec:EquivalenceCheckingTwo} between circuits \( U_1 \) and \( U_2 \) at depth 3 (giving \( U_1^\dagger U_2 \) a depth of 6). The runtime scales approximately linearly with the number of qubits.

For each circuit-equivalence example in the results presented in \Cref{fig: equivalence checking all}, we chose a random circuit $U$ and tested the equivalence of $U$ with itself,
by applying our technique to the circuit for $U^\dagger U$.\footnote{
  The time for the $n=30$ case reported in \Cref{fig:running_time_checking_equivalence} is less than 1 second, which is much less than the $\sim$20-second times reported in \Cref{fig:fixed_depth_changed_n_qubit} for obtaining the local-projection description of a depth-6 circuit.
  The difference is explained by the fact that
  \Cref{fig:running_time_checking_equivalence} reports times for depth-6 circuits of the special form $U^\dagger U$.
  The support of the gates at the last layer of $U$ and the first layer of $U^\dagger$ is the same, which means that in the first layer of $U^\dagger$ the light cone does not become wider;
  consequently, the effective depth of $U^\dagger U$ is 5.
  The time is thus consistent with the $d=5$ case from \Cref{fig:fixed_depth_changed_n_qubit} for the local-projection approach (0.63 seconds). For similar reasons, equivalence/inequivalence checking of depth-4 circuits of the form $U^\dagger U$ only involves circuits of effective depth 7.
  The results are shown in \Cref{fig:depth-4-euivalence checking}.
  The limitations of our hardware platform for computing state vectors does not allow \Cref{fig:depth-4-euivalence checking} to show the crossover point at which the local-projection approach becomes better than the state-vector approach.
}
When performed over complex numbers, the combined circuit implements the identity function.
However, because our implementation uses floating-point arithmetic, and because the matrix multiplications performed for the $U^\dagger$ portion of the circuit are different from those performed for the $U$ portion, floating-point errors naturally arise.
We could have tested the equivalence-checking methods on examples of the form $U_2^\dagger U_1$, where $U_1$ and $U_2$ are different but equivalent circuits (as we did with the micro-benchmarks);
however, examples of the form $U^\dagger U$ provide an experimental control on the possibility that differences between $U_1$ and $U_2$ contribute to floating-point error, revealing empirically the \emph{intrinsic} amount of error that arises with our method.\footnote{
  Equivalence tests of the form $U^\dagger U$ share similarities with the well-known randomized benchmarking technique for hardware quantum gates \cite{Knill_2008}, which can be used to characterize the average error rate caused by noise (while being robust to state-preparation and measurement errors).
  Randomized benchmarking involves applying a sequence $U_1,\cdots, U_m$ of gates chosen randomly from a specified set of quantum gates, followed by the corresponding inverse gates $U_m^{\dag},\cdots, U_1^{\dag}$, and then measuring how well the system returns to the known initial state $\ket{0}^{\otimes n}$.
  Due to hardware noise, denoted by a set of operations $\{ \Lambda_i \}$, one has an ensemble of possible states, so the actual situation before measurement is described by a computation on reduced density matrices $\Lambda_{2m}\mathcal{U}_1^{\dag}\Lambda_{2m-1}\cdots \Lambda_{m+1} \mathcal{U}_m^{\dag}\Lambda_{m} \mathcal{U}_m\Lambda_{m-1}\cdots \mathcal{U}_1\Lambda_0(\op{0^{\otimes n}}{0^{\otimes n}})$, where $\mathcal{U} = \lambda x.U x U^{\dag}$ denotes the quantum operation corresponding to $U$.
  In our experiments, the floating-point errors incurred at each step are similar to the hardware-noise operations $\{ \Lambda_i \}$.
}


\begin{figure}[!tb]
  \begin{minipage}[c]{0.6\textwidth}
    \includegraphics[width=\linewidth]{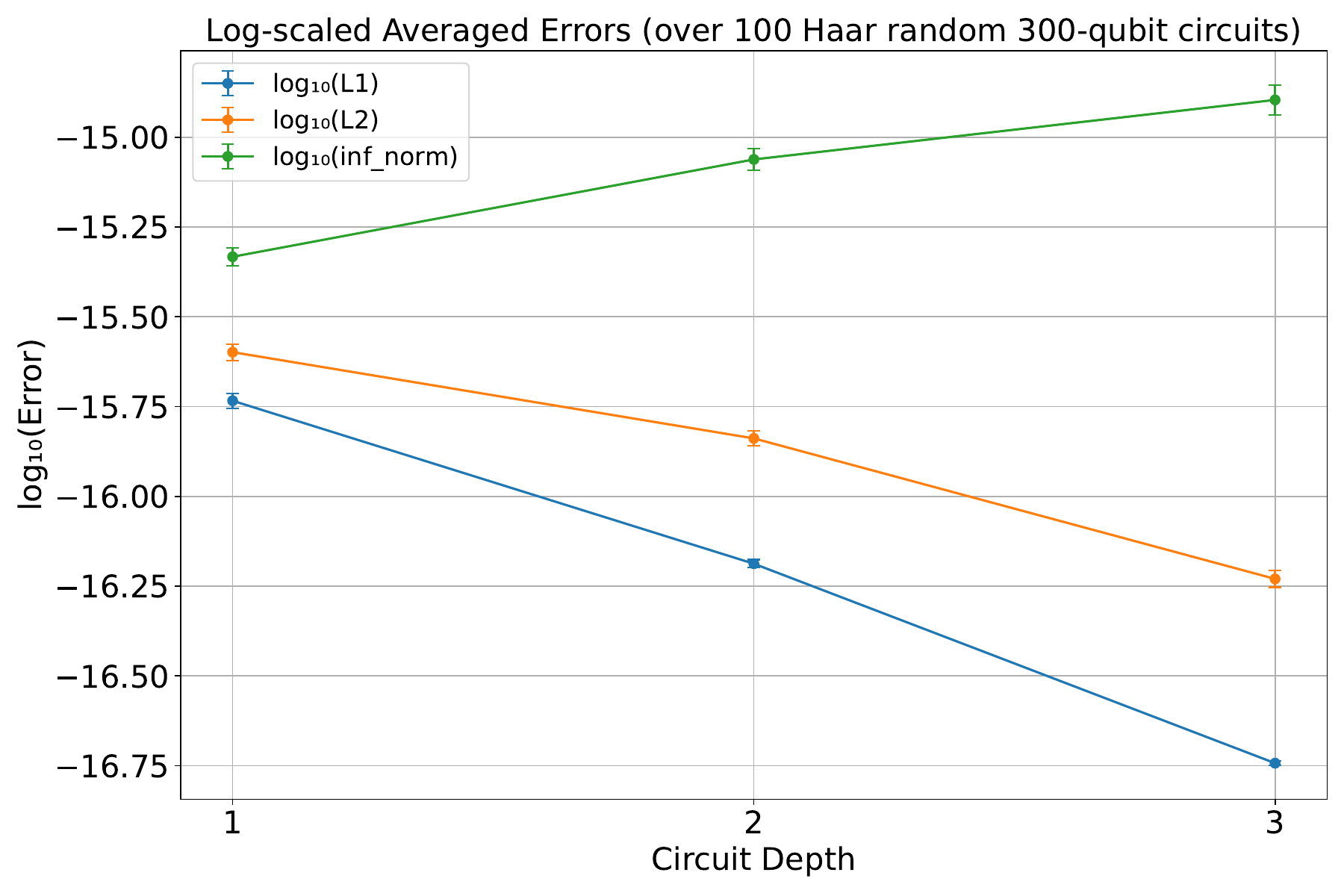}
  \end{minipage}\hfill
  \begin{minipage}[c]{0.37\textwidth}
    \caption{
      Log-scale plots of averages and 1-standard-deviation error bars (visible when the figure is magnified) for the three ``Averaged Error'' measures for a circuit $C$ discussed in the text, for 100 Haar random 300-qubit circuits.
      The depth of $C$ ranged from 1 to 3---and thus the depth of $C^\dagger C$ was 2, 4, or 6.
    } \label{fig:enter-label}
  \end{minipage}
\end{figure}

\textbf{Effects of floating-point error (\Cref{fig:enter-label}).}
For a local-projection matrix $P_{s_i}$ of size $N_i \times N_i$ at the final step of equivalence checking, an $N_i$-dimensional error vector can be
defined as $E = P_{s_i}\ket{0^{\otimes n}}-\ket{0^{\otimes n}}$.
The following standard vector norms provide three error metrics for $P_{s_i}$:
{\small
\begin{align*}
    L_1(E) = \frac{1}{N_i}\sum_{j=1}^{N_i}|E_j|,\ \ L_2(E) = \sqrt{\frac{1}{N_i}\sum_{j=1}^{N_i}|E_j|^2},\ \ L_{\infty} = \max_{j=1,2,...,N_i}|E_j|.
\end{align*}
}
For each circuit $C$, we computed the tuple of local-projection matrices $\mathcal{P} = (P_{s_1},\cdots,P_{s_m})$ for the circuit $C^\dagger C$.
For each matrix $P_{s_i}$, we then computed $E_i = P_{s_i}\ket{0^{\otimes n}}-\ket{0^{\otimes n}}$.
We then computed three norms: $L_1(E_i)$, $L_2(E_i)$, $L_{\infty}(E_i)$, and computed an average over the respective values obtained for $P_{s_i} \in \mathcal{P}$ to obtain three kinds of ``Averaged Error'' measures for circuit $C$.
We take average values in $L_1$ and $L_2$ because for one circuit, local-projection matrices $P_{s_i}$ and $P_{s_j}$ can be of different sizes: $N_i \times N_i$ versus $N_j \times N_j$, respectively.
Average values provide an aggregate measure of how far off the approximate local projections are from the exact solution.

\paragraph{Findings}
The experiments demonstrate the efficiency of local-projection descriptions for shallow quantum circuits as a representation to allow equivalence and inequivalence queries to be answered using a classical computer.
For both kinds of problems, the local-projection-based method exhibits roughly linear growth in run time as a function of number of qubits, whereas the state-vector approach exhibits exponential growth.

In our experiments with equivalence-checking problems of the form $U^\dagger U$, designed to elucidate the intrinsic amount of floating-point error incurred by our method, on average, the ``Average Error'' measure of floating-point error was miniscule---i.e., $\leq {10}^{-15}$.


\section{Related Work}
\label{Se:RelatedWork}


We first compare our method with quantum abstract interpretation \cite{YP21}.
Both methods address the scalability challenge posed by the exponential complexity of classical representations of quantum systems.
At a high level, our approach shares the same foundational idea used by \citeauthor{YP21}: use tuples of local projections to analyze quantum programs, offering a potential pathway to overcome the exponential barrier. 
However, our method differs significantly in several ways and provides multiple advantages:
it automatically selects/adjusts the representation domain, whereas 
\citeauthor{YP21} use a fixed domain;
it is much faster;
it is complete
for a well-defined class of predicates;
and it has applications for checking quantum-circuit equivalence and run-time assertion checking.
For a more detailed comparison, see \Cref{Se:ComparisonWithQuantumAbstractInterpretation}.

{Only a few pieces of previous work prove results about a reasoning method for quantum computing that is complete for some class of assertions.
Examples include Ying \cite{Ying11} and Zhou et al.\ \cite{ZYY19}.
These completeness results rely on classical computation with exponential cost, using $n$-qubit positive semi-definite matrices and projections as predicates, respectively. Their scalability for specific circuit classes, such as shallow circuits, remains unclear.}

Ji and Wu proved a strong result: it is QMA-hard to determine whether a given shallow circuit $U$ satisfies $D(U, I) < a$ or $D(U, I) > b$ for the diamond norm $D(\cdot,\cdot)$, $0 < a < b$ with $b - a > 1/\text{poly}(n)$ and some polynomial \cite{ji2009nonidentitycheckremainsqmacomplete}. Our equivalence-checking result presents a surprising counterpart: one can efficiently distinguish $D(U, I) = 0$ from $D(U, I) > b$ even with a classical computer.


{\citet{10.1145/3618260.3649638} proposed a quasipolynomial-time classical algorithm for approximate sampling from peaked constant-depth circuits. \citet{PhysRevX.12.021021} showed that some constant-depth 2D random circuits allow linear-time approximate simulation despite their universality. However, these results do not apply to our exact equivalence setting.}


For runtime analysis, \citet{li2019proq} used general projections as assertions and proposed computing local projections from the obtained general projections. The main bottleneck there is the efficient representation of general projections as assertions and the imprecise assertion checking due to the non-commutativity of local projections. Our result overcomes these issues for shallow circuits.
\section{Conclusion}
\label{sec:Conclusion}

The paper describes how to create a constraint-based description of the output state of a constant-depth quantum circuit, and how to apply this result to two kinds of equivalence-checking problems and two kinds of assertion-checking problems.
The results are surprising because, in general, constant-depth quantum circuits cannot be accurately simulated on a classical computer \cite{terhal2004adaptivequantumcomputationconstant}.

Although our method represents the first viable approach to solving significant verification challenges that are otherwise intractable with classical computation, it may still incur high costs in some cases. Integrating our approach with recent advances in symbolic quantum simulation using
variants of Binary Decision Diagrams (BDDs) \cite{DBLP:journals/tcad/ZulehnerW19,Book:ZW2020,DBLP:journals/quantum/VinkhuijzenCEDL23,10.1145/3651157} might greatly enhance performance for circuits built from standard gate sets (e.g., Clifford + T)---possibly allowing our techniques to be applied to constant-depth circuits of much greater depth than the ones used in the experiments in \Cref{sec:Experiments}.
Our results pave the way for several other directions of future research, including
\begin{enumerate}
  \item 
    \label{It:CombineWithQAI}
    combining our approach with quantum abstract interpretation (QAI) \cite{YP21} to enable precise, automated assertion checking for general quantum circuits
  \item 
    \label{It:LanguageExtensions}
    extending the approach to handle programs with limited measurements, conditionals, loops and Clifford + T gates \cite{DBLP:journals/pacmpl/YuanC24}
  \item
    \label{It:Noise}
    developing techniques for robust reasoning in the presence of noise, particularly for NISQ-era applications \cite{Wu19,DBLP:conf/pldi/TaoSYHCG21,DBLP:journals/pacmpl/YuanMC22}
  \item 
    \label{It:OptimizationApplication}
    applying our method to 
    support
    quantum-circuit optimizers and certified compilation, including faster circuit-equivalence verification \cite{10.1145/3519939.3523433} and verified optimizers \cite{DBLP:journals/pacmpl/Hietala0HW021}, to enhance the scalability and efficiency of quantum program analysis.
\end{enumerate}

To enable scalable reasoning beyond shallow circuits (item (\ref{It:CombineWithQAI})), we propose partitioning the circuit into two parts, \( C_1 \) and \( C_2 \), where \( C_1 \) remains shallow enough for our algorithm to handle efficiently.
The verification method would proceed as follows:  
\begin{itemize}
    \item
      Apply our algorithm to \( C_1 \), generating local projections that serve as the input predicate.  
    \item
      Use this predicate to perform QAI on \( C_2 \).  
\end{itemize}
This approach
could perform better than
directly applying QAI for two reasons:  
\begin{itemize}
  \item \textbf{Domain selection:}
    Our method automatically determines the ``reasoning domain''---i.e., specific tuples of local projections;
    QAI lacks an automatic mechanism for domain selection.
  \item
    \textbf{Precision of the predicate:} The predicate computed for \( C_1 \) is precise, and monotonicity properties ensure that this method yields better results than directly applying QAI.
\end{itemize}

For quantum while programs (item (\ref{It:LanguageExtensions})), we would need to identify efficient methods to compute loop invariants for while loops.
It might be possible to address this problem
by developing a (terminating) analysis of increasing chains of tuples of local projections.

This paper uses local projections as assertions for quantum circuits of general 2-qubit gates, whereas Clifford circuits rely on the Pauli basis for assertions.
To integrate these two techniques for reasoning about Clifford circuits with T-gates,
it might be possible
to develop a hybrid assertion language that seamlessly combines both approaches.
We note that there are some common parts where the local Pauli observables correspond to local projections \cite{10.1145/3656419}.

\bibliography{main}
\newpage

\appendix

\section{Proof of Theorem \ref{main}.}
\label{sec:ProofOfTheorem}

Theorem \ref{main}. 
For a shallow circuit output $|\psi_{\text{final}}\rangle = \left( \prod_{k=1}^{L} \bigotimes_{(i,j) \in \text{pairs}[k]} U_{ij}^{(k)} \right) |0^{\otimes n}\rangle$ with constant $L$, Algorithm \ref{alg:description}
outputs a tuple of local projections $(P_{s_1},\cdots,P_{s_m})$ that uniquely identifies $|\psi_{\text{final}}\rangle $ in the following sense:
\begin{align*}
\cap_{t=1}^n (P_{s_t} \otimes I_{[n]\setminus s_t})=\{\lambda\ket{\psi_{\text{final}}}| \lambda\in \mathbb{C}\}.
\end{align*} 
Furthermore, each set $s_t$ contains at most constant qubits, and the local projections $P_{s_1},\cdots,P_{s_m}$ pairwise commute. The execution of Algorithm \ref{alg:description} is linear in the depth. Specifically, the loop body needs to be executed \(L\) times for a depth-\(L\) circuit. For a shallow circuit, where \(L\) is a constant, the time complexity is linear in \(n\) because each iteration of the loop body takes time linearly in \(n\).

\begin{proof}
We prove the statement by induction on \(L\).

Algorithm \ref{alg:description} is deterministic, meaning its output depends solely on the input.

\paragraph{Base case:}
For a circuit with \(L=0\), we know that the system state is \(\ket{0}^{\otimes n}\). Algorithm \ref{alg:description} outputs \((P_{s_1}, \ldots, P_{s_m})\) with \(s_t = \{t\}\) and \(P_{s_t} = \op{0}{0}\). Each \(s_t\) contains only one qubit, and \(P_{s_t}\) uniquely identifies the state \(\ket{0}^{\otimes n}\).
The local projections $P_{s_1},\cdots,P_{s_m}$ pairwise commute.

\paragraph{Inductive step:}
Assume that the statement is true for any depth-\(L\) circuit. That is, for any depth-\(L\) circuit with output state \(\ket{\Phi}\), the algorithm will output \((P_{s_1}, \ldots, P_{s_m})\) such that \(|s_t| \leq 2^L\) and

\begin{align*}
\bigcap_{t=1}^n (P_{s_t} \otimes I_{[n] \setminus s_t}) = \{\lambda\ket{\Phi} \mid \lambda \in \mathbb{C}\}
\end{align*}
with the property that the local projections $P_{s_1},\cdots,P_{s_m}$ pairwise commute.

Let us prove the statement for a depth-\(L+1\) circuit. We apply Algorithm \ref{alg:description} on the circuit \(\left( \prod_{k=1}^{L+1} \bigotimes_{(i,j) \in \text{pairs}[k]} U_{ij}^{(k)} \right)\).

Consider the for loop just before the last iteration of ``\(\text{for } k \gets 1 \text{ to } L+1 \text{ do}\).''
The algorithm performs the analysis exactly as in the case of the depth-\(L\) circuit
\(\left( \prod_{k=1}^{L} \bigotimes_{(i,j) \in \text{pairs}[k]} U_{ij}^{(k)} \right)\)
on input \(\ket{0}^{\otimes n}\). Therefore, we can apply induction on \(k = L\). At that time, we know
 \((P_{s_1}^{(L)}, \ldots, P_{s_m}^{(L)})\) such that \(|s_t| \leq 2^L\) and
\begin{align*}
\bigcap_{t=1}^n (P_{s_t}^{(L)} \otimes I_{[n] \setminus s_t}) = \{\lambda\ket{\Phi_L} \mid \lambda \in \mathbb{C}\}
\end{align*}
where $\ket{\Phi_L}:=\left( \bigotimes_{k=1}^{L} \prod_{(i,j) \in \text{pairs}[k]} U_{ij}^{(k)} \right)\ket{0}^{\otimes n}$.

Now consider the final iteration of ``\(\text{for } k \gets 1 \text{ to } L+1 \text{ do}\)''---that is, when \(k=L+1\).
The \(L+1\)-layer unitary \(\bigotimes_{(i,j) \in \text{pairs}[k]} U_{ij}^{(L+1)}\) is considered. In this iteration, for each \(1 \leq t \leq n\), the algorithm checks \(U_{ij}^{(L+1)}\): if \(\{i, j\}\) and \(s_t\) have a nonempty intersection, this unitary will be marked. After iterating over all \(U_{ij}^{(L+1)}\), the algorithm updates \(s_t\) by including all qubits \(\{i, j\}\) with marked \(U_{ij}^{(L+1)}\), and applies the marked unitary on \(P_{s_t}\) to update it.

We first observe that in this single iteration, the number of qubits in \(s_t\) is at most doubled. The reason is as follows: In this layer, the two-element sets \(\{i, j\}\) are disjoint. For \(s_t\), there are at most \(|s_t|\) two-element sets \(\{i, j\}\) that have a nonempty intersection with \(s_t\).
As we union with \(s_t\) all qubits in the \(\{i, j\}\) sets that have a nonempty intersection with \(s_t\), each pair \(\{i, j\}\) brings at most one new element into \(s_t\).
Therefore, after the iteration, the number of qubits in \(s_t\) is at most doubled. By induction, we know that in the output \(|s_t| \leq 2^{L+1}\).

According to
\begin{align*}
\bigcap_{t=1}^n (P_{s_t}^{(L)} \otimes I_{[n] \setminus s_t}) = \{\lambda\ket{\Phi_L} \mid \lambda \in \mathbb{C}\},
\end{align*}
we have
\begin{align*}
\bigotimes_{(i,j) \in \text{pairs}[k]} U_{ij}^{(L+1)}\bigcap_{t=1}^n (P_{s_t}^{(L)} \otimes I_{[n] \setminus s_t}) {\bigotimes_{(i,j) \in \text{pairs}[k]} U_{ij}^{(L+1)}}^{\dag}= \bigotimes_{(i,j) \in \text{pairs}[k]} U_{ij}^{(L+1)}\{\lambda\ket{\Phi_L} \mid \lambda \in \mathbb{C}\}
\end{align*}
By Lemma \ref{lem:intersection}, we know 
\begin{align*}
\bigcap_{t=1}^n [\bigotimes_{(i,j) \in \text{pairs}[k]} U_{ij}^{(L+1)} (P_{s_t}^{(L)} \otimes I_{[n] \setminus s_t}) {\bigotimes_{(i,j) \in \text{pairs}[k]} U_{ij}^{(L+1)}}^{\dag}]= \{\lambda\ket{\Phi_{L+1}} \mid \lambda \in \mathbb{C}\}
\end{align*}
where $\ket{\Phi_{L+1}}:=\bigotimes_{(i,j) \in \text{pairs}[k]} U_{ij}^{(L+1)}\ket{\Phi_{L}}$ is the output state of the $L+1$-depth circuit.
If $\{i,j\}$ and $s_t$ are disjoint, the corresponding $U_{i,j}^{(L+1)}$ commutes with $P_{s_t}^{(L)} \otimes I_{[n] \setminus s_t}$, 
\begin{align*}
&\bigotimes_{(i,j) \in \text{pairs}[k]} U_{ij}^{(L+1)} (P_{s_t}^{(L)} \otimes I_{[n] \setminus s_t}) {\bigotimes_{(i,j) \in \text{pairs}[k]} U_{ij}^{(L+1)}}^{\dag}\\
=&\bigotimes_{\{i,j\} \cap s_t\neq \emptyset} U_{ij}^{(L+1)} (P_{s_t}^{(L)} \otimes I_{[n] \setminus s_t}) {\bigotimes_{\{i,j\} \cap s_t\neq \emptyset} U_{ij}^{(L+1)}}^{\dag} \bigotimes_{\{i,j\} \cap s_t= \emptyset} U_{ij}^{(L+1)}
{\bigotimes_{(i,j) \cap s_t= \emptyset} U_{ij}^{(L+1)}}^{\dag}\\
=&\bigotimes_{\{i,j\} \cap s_t\neq \emptyset} U_{ij}^{(L+1)} (P_{s_t}^{(L)} \otimes I_{[n] \setminus s_t}) {\bigotimes_{\{i,j\} \cap s_t\neq \emptyset} U_{ij}^{(L+1)}}^{\dag}\otimes I
\end{align*}
We use $UU^{\dag}=I$ in the last step.
After employing the updating of $s_t$ in the output, this is exactly
\begin{align*}
\bigcap_{t=1}^n (P_{s_t}^{(L+1)} \otimes I_{[n] \setminus s_t}) = \{\lambda\ket{\Phi_{L+1}} \mid \lambda \in \mathbb{C}\}.
\end{align*}
The property that the local projections $P_{s_1},\cdots,P_{s_m}$ pairwise commute follows from definition.
We have completed the proof for $k=L+1$.

\end{proof}
\section{Appendix: Comparison with Quantum Abstract Interpretation}
\label{Se:ComparisonWithQuantumAbstractInterpretation}


This section compares our method with the approach for quantum abstract interpretation presented by \citet{YP21}.
{
Both their method and ours
}
aim to address the scalability challenge posed by the exponential complexity of classical representations of quantum systems.

At a high level, our approach shares the same foundational idea used by \citeauthor{YP21}: use tuples of local projections to analyze quantum programs. The rationale for this choice is straightforward---it offers a pathway to potentially overcome the exponential barrier. However, the results we achieve differ significantly from those of \citeauthor{YP21}. In the following, we outline these differences across several dimensions: design, computational efficiency, completeness, and applications.


\subsection{{Fixed Domain} versus Automatically Chosen Domain}
\citeauthor{YP21} fix the abstract domain before
carrying out any reasoning steps, via a tuple of qubit sets $S = (s_1,\cdots,s_m)$.
The inflexible domain must be chosen manually.
As a consequence, the method of \citeauthor{YP21} only provides limited assertions, such as whether a state resides within a two-dimensional subspace spanned by two specific tensor product states.

Our work begins with a tuple of single qubits.
Our tuple's domain and content are updated at each unitary layer.
The domain of our tuple can be generated automatically according to the circuit.
By generalizing how tuples are defined and manipulated, our method can explore many more properties of the circuit's output state.

\subsection{Computational Efficiency}
In the method of \citeauthor{YP21}, each step computes the abstract state for a single one- or two-qubit unitary operation. This method ties the computational cost to the gate complexity of the circuit, which is simply the number of gates in the circuit. Additionally, during the verification, the abstract and concretization functions must calculate the intersection of projections, typically using the Gram–Schmidt process. This process is time-consuming and not very robust
{
to floating-point imprecision.
}

In our work, each unitary layer is processed in a single computational step.
This approach ties the computational cost to the time complexity of the circuit, with the depth representing the time required to execute the circuit.
Furthermore, our computation steps only involve multiplying a unitary by a local projection, making our approach highly efficient.

\subsection{Completeness}

For a quantum circuit, a tuple of local projections, and a program analyzer, there could be two different notions of completeness:
1) The tuple of local projections is complete for describing the circuit's output state. In other words, the circuit's output state is the only state that satisfies each local projection.
2) The analyzer will provide the correct answer about whether the circuit's output satisfies each local projection.

\citeauthor{YP21} do not prove completeness or related results.
1) The two-dimensional- subspace assertions that their method can check cannot provide amplitude information;
therefore, it is not complete in the first sense.
2)
    It is not clear whether their method can always prove that a two-dimensional-subspace assertion holds if the outcome satisfies the assertion.
    
Our approach provides both kinds of completeness results for shallow circuits.
1) It constructs a tuple of local projections that uniquely identifies the circuit's output state.
2) We provide an analyzer that is complete in the second sense for assertions as tuples of local projections.

\subsection{Applications}
The method of \citeauthor{YP21} is an abstraction framework. 
For a given quantum program, their method can generate a tuple of local projections such that the output state satisfies these projections.
This approach is satisfactory for showing that some properties are \emph{not} satisfied, rather than that properties \emph{are} satisfied.
It is not immediately applicable for efficient runtime analysis because the order of executing the local projections matters when they do not commute. 

Our approach can be used to prove shallow-circuit equivalence/inequivalence and run-time analysis thanks to its completeness. 
The local projections generated in our approach are always commuting,
{
which enables the kind of run-time analyses described in \Cref{sec:RunTimeAssertions} and \Cref{sec:NISQDeviceVerification}.
}

\end{document}